\documentclass[conference,letterpaper]{IEEEtran}

\usepackage[utf8]{inputenc}
\usepackage[T1]{fontenc}
\usepackage{url}
\usepackage{ifthen}
\usepackage{cite}
\usepackage[cmex10]{amsmath}

\usepackage{graphicx}
\usepackage{amsfonts}
\usepackage{amsmath}
\usepackage{amssymb}
\usepackage{mathrsfs} %%%%%%%%%%%%%
\usepackage{color}
\usepackage{bm}
\usepackage{latexsym} %%%%%%%%%%%%%

\newtheorem{thm}{Theorem}%[section]

\newtheorem{lem}{Lemma}
\newtheorem{proof}{proof}

\newtheorem{defn}{Definition}
\newtheorem{rem}{Remark}
\newtheorem{exam}{Example}

\onecolumn

\begin{document}

\title{Non-binary Two-Deletion Correcting Codes and Burst-Deletion
Correcting Codes}

\author{\IEEEauthorblockN{Wentu~Song and Kui~Cai}
\IEEEauthorblockA{Science, Mathematics and Technology Cluster\\
Singapore University of Technology and Design, Singapore 487372\\
Email: \{wentu\_song, cai\_kui\}@sutd.edu.sg}}

\maketitle

\begin{abstract}
In this paper, we construct systematic $q$-ary two-deletion
correcting codes and burst-deletion correcting codes, where $q\geq
2$ is an even integer. For two-deletion codes, our construction
has redundancy $5\log n+O(\log q\log\log n)$ and has encoding
complexity near-linear in $n$, where $n$ is the length of the
message sequences. For burst-deletion codes, we first present a
construction of binary codes with redundancy $\log n+9\log\log
n+\gamma_t+o(\log\log n)$ bits $(\gamma_t$ is a constant that
depends only on $t)$ and capable of correcting a burst of at most
$t$ deletions, which improves the Lenz-Polyanskii Construction
(ISIT 2020). Then we give a construction of $q$-ary codes with
redundancy $\log n+(8\log q+9)\log\log n+o(\log q\log\log
n)+\gamma_t$ bits and capable of correcting a burst of at most $t$
deletions.
\end{abstract}

\section{Introduction}

DNA-based data storage has been a hot topic in information theory
society. As deletion/insertion are common in DNA data storage
\cite{Heckel20}, codes correcting such errors have attracted
significant attention in recent years.

It was proved in \cite{Levenshtein65} that the optimal redundancy
of binary $t$-deletion correcting codes is asymptotically between
$t\log n+o(\log n)$ and $2t\log n+o(\log n)$, where $n$ is the
length of the code and the redundancy of a binary code $\mathcal
C$ is defined as $n-\log|\mathcal C|$.\footnote{In this paper, for
any positive real number $x$, $\log_qx$ is the logarithm of $x$
with base $q$, where $q\geq 2$ is a positive integer. If the base
$q=2$, then for simplicity, we write $\log_2x=\log x$.} The
well-known Varshamov-Tenengolts (VT) codes \cite{Varshamov65},
which is defined as
\begin{align*}\text{VT}_a(n)\!=\!\left\{\!(c_1,\ldots,c_n)
\!\in\!\{0,1\}^n\!: \!\sum_{i=1}^nic_i\!\equiv a
~\text{mod}\!~(n+1)\!\right\}\!,\end{align*} is a class of binary
single-deletion correcting codes with asymptotically optimal
redundancy. %, which is $\log n+O(1)$. %A systematic encoding
%algorithm of the VT codes with linear-time complexity was proposed
%in \cite{Abdel98}.
%A direct generalization of the VT construction for multiple
%deletions was considered in \cite{Helberg02}, where the
%coefficients in the parity check equation were replaced by a
%$t$-order recursive sequence. However, the asymptotic rate of the
%resulted codes is strictly smaller than $1$.
Construction of multiple-deletion correcting codes with low
redundancy were considered in
\cite{Brakensiek18}$-\!\!$\cite{Wentu22}. By using the higher
order VT syndromes and the syndrome compression technique
\cite{Sima20-1}, Sima \emph{et al}. constructed a family of
systematic $t$-deletion correcting codes with $4t\log n+o(\log n)$
bits \cite{Sima20}. The method in \cite{Sima20} was improved in
\cite{Wentu22} to give a construction of $t$-deletion correcting
codes with redundancy $(4t-1)\log n+o(\log n)$, which is the best
known result in redundancy. For the special case of $t=2$, an
explicit construction of $2$-deletion correcting codes with
redundancy $4\log n+o(\log n)$ was proposed by Guruswami and H\aa
stad \cite{Gur2020}, which matches the existential upper bound of
the asymptotically optimal codes.

As a special case of deletion errors, a burst of $t$ deletions
$($or a $t$-burst-deletion$)$ refers to $t$ deletions that occur
at consecutive positions. It was proved in \cite{Schoeny2017} that
the redundancy of a $t$-burst-deletion-correcting code is
approximately lower bounded by $\log n+t-1$. Levenshtein
\cite{Levenshtein67} constructed a class of binary codes that can
correct a burst of at most two deletions with asymptotically
optimal redundancy of $\log n+1$. Binary codes capable of
correcting a burst of \emph{exact} $t$ deletions for $t\geq 2$ are
constructed in \cite{Schoeny2017}, which also have an
asymptotically optimal redundancy of $\log n+(t-1)\log\log
n+t-\log t$. In \cite{Lenz20}, binary codes capable of correcting
a burst of \emph{at most} $t$ deletions are constructed, which
also have an asymptotically optimal redundancy of $\log
n+(t(t-1)/2)\log\log n+\gamma_t$, where $\gamma_t$ is a constant
that depends only on
$t$.%\footnote{Note that the class of codes that can correct a
%burst of at most $t$ deletions is a more larger class codes that
%can correct a burst of exact $t$ deletions, as a code of the
%earlier type can correct errors of the latter, but the converse is
%not true in general.}

Besides binary codes, nonbinary deletion correcting codes are also
investigated in the literature. In \cite{Levenshtein02}, it was
shown that the optimal redundancy of a $q$-ary $t$-deletion
correcting code is asymptotically lower bounded by $t\log n+t\log
q+o(\log q\log n)$ and upper bounded by $2t\log n+t\log q+o(\log
q\log n)$ in bits $(q\geq 2)$. A class of $q$-ary single-deletion
correcting codes with redundancy close to the asymptotic
optimality was constructed in \cite{Tenengolts84}. For $q$-ary
$t$-deletion correcting codes, the best known construction is
presented in \cite{Sima20-2}, which achieve optimal redundancy up
to a constant factor. Quaternary codes capable of correcting a
single edit error for DNA data storage were studied in
\cite{Cai19}. In \cite{Wang21}, a $q$-ary code that can correct a
burst of at most $2$ deletions with redundancy $\log n+O(\log q
\log\log n)$ bits was constructed, where $q\geq 2$ is an even
integer.

In this paper, we construct nonbinary two-deletion correcting
codes and burst-deletion correcting codes. Our contributions
includes:
\begin{itemize}
 \item[1)] We construct a class of systematic $q$-ary two-deletion
 correcting codes, with redundancy $5\log n+O(\log q\log\log n)$,
 where $q\geq 2$ is an even integer and $n$ is the length of the
 message sequences.
 \item[2)] We present a construction of binary codes with redundancy $\log
 n+9\log\log n+\gamma_t+o(\log\log n)$ bits $(\gamma_t$ is a
 constant that depends only on $t)$ and capable of correcting a
 burst of at most $t$ deletions, which improves the Lenz-Polyanskii
 Construction (ISIT 2020).
 \item[2)] We give a construction of $q$-ary codes with redundancy
 $\log n+(8\log q+9)\log\log n+o(\log q\log\log n)+\gamma_t$ bits
 and capable of correcting a burst of at most $t$ deletions,
 where $q\geq 2$ is an even integer.
\end{itemize}

Note that each symbol in $\mathbb Z_q$ can be viewed as a binary
string of length $\lceil\log q\rceil$, so a binary code of length
$\lceil\log q\rceil n$ and capable of correcting a burst of
$\lceil\log q\rceil t$ deletions can also be viewed as a $q$-ary
code of length $n$ and capable of correcting a burst of $t$
deletions. By this observation and by the construction in
\cite{Lenz20}, we can obtain a $q$-ary code of length $n$ and
capable of correcting a burst of $t$ deletions that has redundancy
$$\log(n\log q)+\frac{t\log q (t\log q+1)}{2}\log\log(n\log
q)+\gamma_t.$$ Our construction has improved redundancy than this
naive construction.

The rest of this paper is organized as follows. In Section
\uppercase\expandafter{\romannumeral 2}, we introduce some basic
concepts and notations of deletion correcting codes, and review
some related constructions in the literature. In Section
\uppercase\expandafter{\romannumeral 3}, we construct $q$-ary
two-deletion correcting codes. In Section
\uppercase\expandafter{\romannumeral 4}, we present an improved
construction of binary codes correcting a burst of at most $t$
deletions. In Section \uppercase\expandafter{\romannumeral 5}, we
construct of $q$-ary codes correcting a burst of at most $t$
deletions. The paper is concluded in Section
\uppercase\expandafter{\romannumeral 6}.

\section{Preliminaries}

For any integers $m$ and $n$ such that $m\leq n$, we denote
$[m,n]=\{m,m+1,\ldots,n\}$ and call it an \emph{interval}. If
$m>n$, let $[m,n]=\emptyset$. For simplicity, denote $[n]=[1,n]$
for any positive integer $n$. For any positive real number $x$,
$\log x$ is the logarithm of $x$ with base $2$, i.e., $\log
x=\log_2x$. The size (cardinality) of any set $S$ is denoted by
$|S|$. For any positive integer $q\geq 2$, denote $\mathbb
Z_q=\{0,1,2,\cdots,q-1\}$, which will be used as the alphabet of
$q$-ary codes.

For any string (also called a sequence) $\bm{x}\in\mathbb Z_q^n$,
$n$ is called the length of $\bm{x}$ and denote $|\bm x|=n$.
Unless otherwise specified, we use $x_i$ to denote the $i$th
coordinate of $\bm{x}$, where $i\in[n]$. Usually, we denote
$\bm{x}=(x_1,x_2,\ldots,x_n)$ or $\bm{x}=x_1x_2\cdots x_n$. For
any $I=\{i_1, i_2, \ldots, i_d\}\subseteq [n]$ such that
$i_1<i_2<\cdots<i_d$, denote $x_I=x_{i_1} x_{i_2} \cdots x_{i_d}$
and call $x_D$ a \emph{subsequence} of $\bm x$. If $I\subseteq[n]$
is an interval $($i.e., $I=[i,j]$ for some $i,j\in[1,n]$, $i\leq
j)$, then $x_{I}=x_{[i,j]}=x_{i}x_{i+1}\cdots x_{j}$ is called a
\emph{substring} of $\bm x$. In other words, a substring of
$\bm{x}$ is a subsequence of $\bm{x}$ consisting of some
consecutive symbols of $\bm{x}$. We say that $\bm x$ contains $\bm
p~($or $\bm p$ is contained in $\bm x)$ if $\bm{p}$ is a substring
of $\bm{x}$. For two substrings $x_{I}$ and $x_{I'}$ of $\bm x$,
where $I,I'\subseteq[n]$ are two intervals, we say that $x_{I}$
and $x_{I'}$ are \emph{disjoint} if $I\cap I'=\emptyset$.

Let $t\leq n$ be a nonnegative integer. For any $\bm{x}\in\mathbb
Z_q^n$, let $\mathcal D_t(\bm x)$ denote the set of subsequences
of $\bm x$ of length $n-t$, and let $\mathcal B_{t}(\bm x)$ denote
the set of subsequences $\bm y$ of $\bm x$ that can be obtained
from $\bm x$ by a burst of $t$ deletions, that is $\bm y=x_{I}$
such that $I=[n]\backslash D$ for some interval $D\subseteq[n]$ of
length $t~($i.e., $D=[i,i+t-1]$ for some $i\in[n-t+1])$. Moreover,
let $\mathcal B_{\leq t}(\bm x)=\bigcup_{t'=0}^t\mathcal
B_{t'}(\bm x)$ be the set of subsequences of $\bm x$ that can be
obtained from $\bm x$ by a burst of at most $t$ deletions.
Clearly, $\mathcal D_{ 1}(\bm x)=\mathcal B_{1}(\bm x)=\mathcal
B_{\leq 1}(\bm x)$. However, $\mathcal B_{t}(\bm x)\subseteq
\mathcal D_{t}(\bm x)\cap\mathcal B_{\leq t}(\bm x)$ for $t\geq
2$.

A code $\mathcal C\subseteq\mathbb Z_q^n$ is said to be a
$t$-\emph{deletion correcting code} if for any $\bm x\in\mathcal
C$ and any $\bm y\in\mathcal D_t(\bm x)$, $\bm x$ can be uniquely
recovered from $\bm y$; the code $\mathcal C\subseteq\mathbb
Z_q^n$ is said to be capable of \emph{correcting a burst of at
most $t$ deletions} if for any $\bm x\in\mathcal C$ and any $\bm
y\in\mathcal B_{\leq t}(\bm x)$, $\bm x$ can be uniquely recovered
from $\bm y$.

\subsection{Some Constructions Related to Binary Single-deletion
and Two-deletion Correcting Codes}

From the VT construction, we can obtain the following lemma about
single-deletion correcting codes.

\begin{lem}\label{VT-Code-Skch}
For any integer $n\geq 3$, there exists a function $\text{VT}:
\{0,1\}^n\rightarrow\{0,1\}^{\log n}$, computable in linear time,
such that for any $\bm{c}\in\{0,1\}^n$, given $\text{VT}(\bm c)$
and any $\bm b\in\mathcal D_1(\bm c)$, one can uniquely recover
$\bm c$.
\end{lem}

The following lemma can be obtained from the results of
\cite{Sima19-1}, and so its proof is omitted.

\begin{lem}\label{Ind-Bnry-2del}
For any integer $n\geq 3$, there exists a function $\xi:
\{0,1\}^n\rightarrow\{0,1\}^{7\log n+o(\log n)}$, computable in
linear time, such that for any $\bm{c}\in\{0,1\}^n$, given
$\xi(\bm c)$ and any $\bm b\in\mathcal D_2(\bm c)$, one can
uniquely recover $\bm c$.
\end{lem}

Lemma \ref{Ind-Bnry-2del} can be used to construct systematic
binary two-deletion correcting codes with redundancy not greater
than $7\log n+o(\log n)$. Another construction, which uses the
so-called regular strings and has lower redundancy, was proposed
in \cite{Gur2020}, but it is not systematic.

\begin{defn}[Regularity]\label{Regu} A binary string
$\bm{c}\in\{0,1\}^n$ is said to be \emph{regular} if each
(contiguous) sub-string of $\bm c$ of length at least $d\log n$
contains both $00$ and $11$.
\end{defn}

In Definition \ref{Regu}, $d$ is a constant that can be chosen
properly. In this paper, we will always choose $d=7$. The
following two lemmas are from \cite{Gur2020}.

\begin{lem}\cite[Lemma 11]{Gur2020}\label{Enc-Reg-Bnry}
There exist an integer $M\geq 2^{n-1}$ and a one-to-one mapping
$\text{RegEnc}: \{1,2,\cdots,M\}\rightarrow\{0,1\}^n$ such that
its image is contained in the set of regular strings. Moreover,
the function $\text{RegEnc}$ can be computed in near-linear time
with a polynomial size lookup table.
\end{lem}

\begin{lem}\cite[Theorem 7]{Gur2020}\label{Reg-Bnry-2-Del}
There is a function $\eta$, computable in linear time, that maps
$n$ bits to $4\log n+10\log\log n+O(1)$ bits such that for any
regular $\bm{c}\in\{0,1\}^n$, given $\eta(\bm c)$ and any $\bm
b\in\mathcal D_2(\bm c)$, one can uniquely recover $\bm c$.
\end{lem}

\subsection{Some Constructions Related to Binary Burst-Deletion Correcting
Codes}

The following lemma can be obtained from the results in Section IV
of \cite{Sima20-1}.

\begin{lem}\label{lem-Bnry-burst-Sima}
Suppose $t$ is a constant with respect to $n$. There is a function
$\phi:\{0,1\}^n\rightarrow\{0,1\}^{4\log n+o(\log n)}$, computable
in time $O(2^tn^3)$, such that for any $\bm c\in\{0,1\}^n$, given
$\phi(\bm c)$ and any $\bm b\in\mathcal B_{\leq t}(\bm c)$, one
can uniquely recover $\bm c$.
\end{lem}

Let $m\leq\delta\leq n$ be positive integers and $\bm
p\in\{0,1\}^m$, where $\bm p$ is called a \emph{pattern}. A string
$\bm c\in\{0,1\}^n$ is called $(\bm p, \delta)$-\emph{dense}, if
each substring of $\bm c$ of length $\delta$ contains at least one
pattern $\bm p$.

As in \cite{Lenz20}, in this paper, we take
$$\delta=t2^{t+1}\log n\footnotemark{}$$ and
$$\bm p=0^t1^t,$$ where $0^t$ is the string consists of $t$ symbol $0$s, and
$1^t$ is the string consists of $t$ symbol $1$s. In other words,
$\bm p=p_1p_2\cdots p_{2t}$ such that $p_1=p_2=\cdots=p_t=0$ and
$p_{t+1}=p_{t+2}=\cdots=p_{2t}=1$. It was proven in \cite{Lenz20}
that one bit of redundancy is sufficient to construct $(\bm p,
\delta)$-dense string. \footnotetext{In \cite{Lenz20}, $\delta$ is
taken to be $t2^{t+1}\lceil\log n\rceil\footnotemark[2]$. In this
paper, for notational simplicity, we omit the ceiling function and
write $\delta=t2^{t+1}\log n$.}%\setcounter{footnote}{2}

\begin{lem}\cite[Lemma 1]{Lenz20}\label{lem-p-dense}
For any $n\geq 5$, the number of $(\bm p, \delta)$-dense strings
of length $n$ is at least $$2^n(1-n^{1-\log e})\geq 2^{n-1}.$$
\end{lem}

The following lemma can be obtained from Construction 1 and Lemma
2 of \cite{Lenz20} and so its proof is omitted.

\begin{lem}\label{lem-Bnry-burst-Lenz}
For any positive integer $n$, there is a function $\mu$,
computable in linear time, that maps $n$ bits to $\log n+3$ bits
such that for any $(\bm p, \delta)$-dense $\bm c\in\{0,1\}^n$,
given $\mu(\bm c)$ and any $\bm b\in\mathcal B_{\leq t}(\bm c)$,
one can find in time $O(n)$ an interval $L\subseteq[n]$ of length
at most $\delta+t$ such that $\bm b=c_{[n]\backslash D}$ for some
interval $D\subseteq L~($i.e., the deletions are located in the
interval $L)$.
\end{lem}

\subsection{Matrix Representation of $q$-ary Strings}

In the rest of this paper, we always assume $q>2$ is a fixed even
integer. As in \cite{Sima20-2}, each $q$-ary string
$\bm{x}=x_1x_2\ldots x_n\in\mathbb Z_q^{n}$ can be represented by
a $\lceil\log q\rceil\times n$ binary matrix
\begin{align}\label{q-B-Repr}
M_{\bm{x}}=(c_{i,j})=\left(\begin{array}{cccc}
c_{1,1} & \cdots & c_{1,n} \\
\vdots & \ddots & \vdots \\
c_{\lceil\log q\rceil,1} &
\cdots & c_{\lceil\log q\rceil,n} \\
\end{array}\right),
\end{align}
where $c_{i,j}\in\{0,1\}$, such that the $j$th column of
$M_{\bm{x}}$ is the binary representation of $x_j$. Specifically,
$x_j=\sum_{i=1}^{\lceil\log q\rceil}c_{i,j}2^{i-1}$. We call
$M_{\bm{x}}$ the \emph{matrix representation} of $\bm x$. For any
$i\in\{1,2,\cdots,\lceil\log q\rceil\}$ and any interval
$J=[j_1,j_2]=\{j_1,j_1+1,\cdots,j_2\}\subseteq[n]$, where $1\leq
j_1<j_2\leq n$, denote
\begin{align}\label{1-row-Repr}
c_{i,J}\triangleq c_{i,j_1}c_{i,j_1+1}\cdots c_{i,j_2},\end{align}
which is a substring of the $i$th row of $M_{\bm x}$ consisting of
$c_{i,j_1}$, $c_{i,j_1+1}$, $\cdots$, $c_{i,j_2}$. In particular,
$c_{i,[n]}$ is the $i$th row of $M_{\bm x}$.

Clearly, if $\bm y\in\mathbb Z_q^{n-t}$ is obtained from $\bm x$
by deleting $x_{j_1},\cdots,x_{j_t}$, then the matrix
representation $M_{\bm{y}}$ of $\bm y$ can be obtained from
$M_{\bm{x}}$ by deleting columns $j_1, \cdots,j_t$ of
$M_{\bm{x}}$. Moreover, $\bm x$ can be recovered from $\bm y$ if
and only if its matrix representation $M_{\bm{x}}$ can be
recovered from $M_{\bm{y}}$.

\begin{lem}\label{lem-En-B2Q}
Suppose $\mathcal E_{0}:\{0,1\}^{n-1}\rightarrow\{0,1\}^{n}$ is a
one-to-one mapping and $q>2$ is an even integer. Then there is a
one-to-one mapping $\bar{\mathcal E}_{0}:\mathbb
Z_q^{n-1}\rightarrow\mathbb Z_q^{n}$, with the same computing time
as $\mathcal E_{0}$, such that for any $\bm u\in\mathbb Z_q^{n-1}$
and $\bm x=\bar{\mathcal E}_{0}(\bm u)$, if
$M_{\bm{u}}=(b_{i,j})_{\lceil\log q\rceil\times(n-1)}$ and
$M_{\bm{x}}=(c_{i,j})_{\lceil\log q\rceil\times n}$ are the matrix
representation of $\bm u$ and $\bm x$ respectively, then
$$c_{1,[n]}=\mathcal E_{0}(b_{1,[n-1]}).$$
\end{lem}
\begin{proof}
For each $\bm u\in\mathbb Z_q^{n-1}$, where the matrix
representation of $\bm u$ is
\begin{align*}
M_{\bm{u}}=\left(\begin{array}{ccccc}
b_{1,1} & \cdots & b_{1,n-1} \\
b_{2,1} & \cdots & b_{2,n-1} \\
\vdots & \ddots & \vdots \\
b_{\lceil\log q\rceil,1} &
\cdots & b_{\lceil\log q\rceil,n-1} \\
\end{array}\right),
\end{align*}
denote $\mathcal E_{0}(b_{1,[n-1]})=\bm c=c_{1,1}\cdots
c_{1,n-1}c_{1,n}$ and let
\begin{align*}
M=\left(\begin{array}{ccccc}
c_{1,1} & \cdots & c_{1,n-1} & c_{1,n} \\
b_{2,1} & \cdots & b_{2,n-1} & 0 \\
\vdots & \ddots & \vdots & \vdots \\
b_{\lceil\log q\rceil,1} &
\cdots & b_{\lceil\log q\rceil,n-1} & 0 \\
\end{array}\right).
\end{align*}
Specifically, $M=(c_{i,j})$ is a $\lceil\log q\rceil\times n$
binary matrix satisfying the following three properties: i) the
first row of $M$ is equal to $\bm c$; ii) $c_{2,j}\cdots
c_{\lceil\log q\rceil,j}=b_{2,j}\cdots b_{\lceil\log q\rceil,j}$
for each $j\in[n-1]$; iii) $c_{2,n}\cdots c_{\lceil\log
q\rceil,n}=0^{\lceil\log q\rceil-1}$, where $0^{\lceil\log
q\rceil-1}$ is the string consisting of $\lceil\log q\rceil-1$
symbol $0$s.

Let $\bar{\mathcal E}_{0}(\bm u)=\bm x$ such that the matrix
representation of $\bm x$ is $M_{\bm{x}}=M$. It is easy to see
that $c_{1,[n]}=\mathcal E_{0}(b_{1,[n-1]})$ and the computing
time of $\bar{\mathcal E}_{0}$ is the same as that of $\mathcal
E_{0}$. Moreover, since $\mathcal E_{0}$ is a one-to-one mapping,
it is also easy to see that $\bar{\mathcal E}_{0}$ is a one-to-one
mapping.

It remains to prove that $\bm x\in\mathbb Z_q^{n}$, equivalently,
each column of $M_{\bm{x}}$ is the binary representation of some
integer in $\mathbb Z_q$.

According to property iii) of the constructed matrix $M$, we have
$c_{\lceil\log q\rceil,n}\cdots c_{2,n}c_{1,n}=0^{\lceil\log
q\rceil-1}c_{1,n}$, so the last column of $M$ is the binary
representation of $c_{1,n}\in\{0,1\}\subseteq\mathbb Z_q$. For
each $j\in[n-1]$, according to property ii) of $M$, we have
$c_{\lceil\log q\rceil,j}\cdots c_{2,j}c_{1,j}=b_{\lceil\log
q\rceil,j}\cdots b_{2,j}c_{1,j}$, so $x_j=\sum_{i=1}^{\lceil\log
q\rceil}c_{i,j}2^{i-1}=\sum_{i=2}^{\lceil\log
q\rceil}b_{i,j}2^{i-1}+c_{1,j}=u_{j}-b_{1,j}+c_{1,j}$, where the
last equality holds because according to the definition of the
matrix representation, $b_{\lceil\log q\rceil,j}\cdots
b_{2,j}b_{1,j}$ is the binary representation of $u_j$. If
$b_{1,j}=1$, then $x_j=u_{j}-b_{1,j}+c_{1,j}\leq u_j\leq q-1$. If
$b_{1,j}=0$, then $u_{j}$ is even. Noticing that $q$ is even, so
$u_{j}\leq q-2$, and hence
$x_j=u_{j}-b_{1,j}+c_{1,j}=u_{j}+c_{1,j}\leq q-1$. In both cases,
we have $x_j\in\mathbb Z_q$. Thus, each column of $M$ is the
binary representation of some integer in $\mathbb Z_q$, and so
$\bm x\in\mathbb Z_q^{n}$.
\end{proof}

\section{Nonbinary Two-deletion Correcting Codes}

In this section, we consider $q$-ary two-deletion correcting
codes. We assume that $q>2$ is an even integer and is a constant
with respect to the code length $n$. Each binary sequence $\bm a$
will also be viewed as a non-negative integer whose binary
representation is $\bm a$, and conversely, each non-negative
integer $m$ will also be viewed as a binary sequence with length
$\lceil\log (m+1)\rceil$, i.e., the binary representation of $m$.
Therefore, summation and multiplication of binary strings and
integers are performed in the set of integers.

We need to introduce some concepts and notations for binary
strings, which will be used in our construction.

Let $\bm c\in\{0,1\}^n$ be a binary string of length $n$. A
\emph{run} of $\bm c$ is a maximal substring of $\bm c$ consisting
of identical symbols.\footnote{We say that a substring of $\bm c$
satisfying a certain property is maximal if it is contained by no
other substring of $\bm c$ that satisfies the same property.
Hence, a maximal run of the string $\bm c$ is not contained by any
other run of $\bm c$.} A substring $c_{[i_1,i_2]}$ of $\bm c$,
where $i_1<i_2$, is called an \emph{alternative substring} of $\bm
c$ if $c_{i+1}\neq c_{i}$ for all $i\in[i_1,i_2-1]$.

\begin{rem}\label{rem-Regu-length}
From Definition \ref{Regu}, it is easy to see that if
$\bm{c}\in\{0,1\}^n$ is regular, then each substring of $\bm c$ of
length $d\log n$ can not be a run or an alternative substring of
$\bm c$ because it contains both $00$ and $11$. Equivalently, each
run and each alternative substring of $\bm c$ have length at most
$d\log n$.
\end{rem}

\begin{defn}\label{def-run-set}
For each $\bm c\in\{0,1\}^n$, let $c_{I_i}$ be the $i$th run
$($counting from the left$)$ of $\bm c$, where $I_i\subseteq[n]$
is the index set of $c_{I_i}$. Then we denote $\mathcal I_{\bm
c}=\{c_{I_1},\cdots,c_{I_{n'}}\}$ and call it \emph{the set of
runs} of $\bm c$, where $n'$ is the number of runs of $\bm c$.
\end{defn}

Let $\text{VT}$, $\xi$ and $\eta$ be the functions constructed by
Lemma \ref{VT-Code-Skch}, Lemma \ref{Ind-Bnry-2del} and Lemma
\ref{Reg-Bnry-2-Del}, respectively. Denote $$\rho=3d\log n$$ and
let
\begin{equation}\label{def-Ji-intvl}
J_j=\!\left\{\!\begin{aligned} &[(j-1)\rho\!+1, (j+1)\rho],
~\text{for}~j\in\!\{1,\cdots, \left\lceil n/\rho\right\rceil-2\},\\
&[(j-1)\rho\!+1, n], ~~~~~~~~~\text{for}~j=\left\lceil
n/\rho\right\rceil-1.
\end{aligned}\right.
\end{equation}
Note that each interval $J_j$ has length $2\rho$ and the
intersection of two successive intervals $J_{j}$ and $J_{j+1}$ is
an interval of length $\rho$. It is easy to see the following
remark.
\begin{rem}\label{rem-sets-Ji}
The intervals $J_j$, $j=1,\cdots,\left\lceil n/\rho\right\rceil-1$
satisfies:
\begin{itemize}
 \item[1)] For any interval $J\subseteq[n]$ of length at most
 $\rho$, we can find an
 $j_0\in\{1,2,\cdots,\left\lceil n/\rho\right\rceil-1\}$
 such that $J\subseteq J_{j_0}$.
 \item[2)] $J_j\cap J_{j'}=\emptyset$ for all $j,j'\in\{1,2,\cdots,
 \left\lceil n/\rho\right\rceil-1\}$ such that $|j-j'|\geq 2$.
\end{itemize}
\end{rem}

For each $q$-ary string $\bm x\in\mathbb Z_q^n$, let $M_{\bm
x}=(c_{i,j})$ be the matrix representation of $\bm x$ as defined
by \eqref{q-B-Repr} and $c_{1,[n]}$ be the first row of $M_{\bm
x}$. We construct a function $f$ as follows.

\textbf{Construction 1}: For each $\bm x\in\mathbb Z_q^n$, let
$\mathcal I_{\bm c}=\{c_{I_1},\cdots,c_{I_{n'}}\}$ be the set of
runs of $\bm c=c_{1,[n]}$ as defined in Definition
\ref{def-run-set}. For each $i\in[n']$, let
$$g_i(\bm x)=\left(\text{VT}(c_{2,I_i}), \text{VT}(c_{3,I_i}),
\cdots, \text{VT}(c_{\lceil\log q\rceil,I_i})\right),$$ and for
each $\ell\in\{0,1\}$, let
\begin{align}\label{def-2del-g-ell}
g^{(\ell)}(\bm x)=\sum_{i=1}^{n'}i^\ell
g_i(\bm x)~\text{mod}~2n^\ell N_1,\end{align} where
$$N_1=q^{\log\log n+3}.$$ Moreover, for each
$j\in\{1,\cdots,\left\lceil n/\rho\right\rceil-1\}$, let
$$h_j(\bm x)=\left(\xi(c_{2,J_j}),\xi(c_{3,J_j}),\cdots,\xi(c_{\lceil\log
q\rceil,J_j})\right),\footnotemark{}$$ and for each
$\ell\in\{0,1\}$, let
\begin{align}\label{Constr1-h-ell}
h^{(\ell)}(\bm x)=\sum_{\substack{j\in\{1,\cdots,\left\lceil
n/\rho\right\rceil-1\}:\\
j\!~\equiv\!~\ell~\text{mod}~2}} h_j(\bm
x)~\text{mod}~N_2\end{align} where $$N_2=q^{7\log\log n+o(\log\log
n)}.$$ Finally, let
$$f(\bm x)=\left(\eta(c_{1,[n]}),g^{(0)}(\bm x),g^{(1)}(\bm x),
h^{(0)}(\bm x),h^{(1)}(\bm x)\right).$$ \footnotetext{Note that
Lemma \ref{Ind-Bnry-2del} requires that each $|I_j|\geq 3$. If
$|I_j|<3$, we can just let $\xi(c_{i,I_j})=c_{i,I_j}$. Then
$c_{i,I_j}$ can also be recovered from $\xi(c_{i,I_j})$. This is
feasible because in our construction, we only need that each
$\xi(c_{i,I_j})$ is a sequence of length not greater than
$7\log\log n+o(\log\log n)$.}%\setcounter{footnote}{5}

Let $\mathcal R_n$ denote the set of all $\bm x\in\mathbb Z_q^{n}$
such that $c_{1,[n]}$ is a regular string with $d=7~($according to
Definition \ref{Regu}$)$. Then we have the following Theorem.

\begin{thm}\label{thm-2del-sketch}
The function $f(\bm x)$ is computable in linear time and the
length $|f(\bm x)|$ of $f(\bm x)$ satisfies $$|f(\bm x)|\leq 5\log
n+O(\log q\log\log n).$$ Moreover, if $\bm{x}\in\mathcal R_n$,
then $\bm x$ can be uniquely recovered from $f(\bm x)$ and any
given $\bm y\in\mathcal D_2(\bm x)$.
\end{thm}

To prove Theorem \ref{thm-2del-sketch}, we need the following
lemma.

\begin{lem}\label{lem-2del-pstn}
Suppose $\bm c\in\{0,1\}^n$ is regular and $\bm b\in\{0,1\}^{n-2}$
such that $\bm b$ can be obtained from $\bm c$ by deleting two
symbols of $\bm c$. Then exact one of the following holds.
\begin{itemize}
 \item[1)] There are two distinct runs $c_{I_{j_1}}$ and $c_{I_{j_2}}$ of $\bm
 c$, uniquely determined by $\bm b$ and $\bm c$,
 such that $\bm b$ can be obtained from $\bm c$ by deleting one symbol
 in $c_{I_{j_1}}$ and one symbol in $c_{I_{j_2}}$.
 \item[2)] There is an interval $J\subseteq[n]$ of length at most
 $\rho$ such that $\bm b$ can be obtained from $\bm c$ by deleting two
 symbols in $c_{J}$.
\end{itemize}
\end{lem}
\begin{proof}
This Lemma is proved in Appendix A.
\end{proof}

Now, we can prove Theorem \ref{thm-2del-sketch}.

\begin{proof}[Proof of Theorem \ref{thm-2del-sketch}]
Note that by Lemma \ref{VT-Code-Skch}, Lemma \ref{Ind-Bnry-2del}
and Lemma \ref{Reg-Bnry-2-Del}, the functions $\text{VT}$, $\xi$
and $\eta$ are all computable in linear time. By Construction 1,
the functions $g^{(\ell)}(\bm x)$ and $h^{(\ell)}(\bm x)$,
$\ell\in\{0,1\}$, are computable in linear time. Hence, $f(\bm
x)=\left(\eta(c_{1,[n]}),g^{(0)}(\bm x),g^{(1)}(\bm x),
h^{(0)}(\bm x),h^{(1)}(\bm x)\right)$ is computable in linear
time.

For each $\bm x\in\mathbb Z_q^n$, by Construction 1, the length
$|g^{(\ell)}(\bm x)|$ of $g^{(\ell)}(\bm x)$, $\ell\in\{0,1\}$,
satisfies
\begin{align*}|g^{(\ell)}(\bm x)|&\leq \log (2n^\ell
N_1)\\&=\ell\log n+\log q\log\log n+1\end{align*} Similarly, the
length $|h^{(\ell)}(\bm x)|$ of $h^{(\ell)}(\bm x)$,
$\ell\in\{0,1\}$, satisfies
\begin{align*}|h^{(\ell)}(\bm x)|&\leq \log N_2\\
&=\log q(7\log\log n+o(\log\log n))\end{align*} Moreover, by Lemma
\ref{Reg-Bnry-2-Del}, the length of $\eta(c_{1,[n]})$ satisfies
$$|\eta(c_{1,[n]})|\leq4\log n+10\log\log n+O(1).$$ Thus, by
Construction 1, the length of $f(\bm x)$ satisfies
\begin{align*}|f(\bm x)|&=|\eta(c_{1,[n]})|
+|g^{(0)}(\bm x)|+|g^{(1)}(\bm x)|\\&~~~~+|h^{(0)}(\bm
x)|+|h^{(1)}(\bm x)|\\&\leq 5\log n+(16\log q+10)\log\log
n\\&~~~~+o(\log q\log\log n)\\&=5\log n+O(\log q\log\log
n).\end{align*}

It remains to prove that for each $\bm{x}\in\mathcal R_n$, given
$f(\bm x)$ and any $\bm y\in\mathcal D_2(\bm x)$, one can uniquely
recover $\bm x$. To prove this, we first prove that $g_j(\bm
x)<N_1$ for each $j\in[n']$ and $h_j(\bm x)<N_2$ for each
$j\in\{1,2,\cdots,\left\lceil n/\rho\right\rceil-1\}$.

Since $\bm{x}\in\mathcal R_n$, by Remark \ref{rem-Regu-length},
each run of $c_{1,[n]}$ has length at most $7\log n~($noticing
that we take $d=7$ in this paper$)$, so for each
$i\in[2,\left\lceil\log q\right\rceil]$ and $j\in[n']$, we have
$c_{i,I_j}\in\{0,1\}^{\leq 7\log n}$. By Lemma \ref{VT-Code-Skch},
for each $j\in[2,\left\lceil\log q\right\rceil]$, the length of
$\text{VT}(c_{j,I_i})$ satisfies $|\text{VT}(c_{j,I_i})|\leq
\log(7\log n)\leq\log\log n+3.$ By Construction 1, the length of
$g_j(\bm x)$ satisfies \begin{align*}|g_j(\bm x)|&\leq
(\left\lceil\log q\right\rceil-1)(\log\log n+3)\\&<(\log
q)(\log\log n+3),\end{align*} so
$$g_j(\bm x)<q^{\log\log n+3}=N_1.$$
Similarly, for each $i\in[2,\left\lceil\log q\right\rceil]$ and
each $j\in\{1,2,\cdots,\left\lceil n/\rho\right\rceil-1\}$, by
\eqref{def-Ji-intvl}, the length of the interval $c_{i,J_j}$
satisfies $|c_{i,J_j}|\leq 2\rho=42\log n$, so by Lemma
\ref{Ind-Bnry-2del}, we have
\begin{align*}|\xi(c_{i,I_j})|&\leq 7\log(42\log n)+o(\log(42\log
n))\\&=7\log\log n+o(\log\log n).\end{align*} By Construction 1,
\begin{align*}|h_j(\bm x)|&\leq(\left\lceil\log q\right\rceil-1)(7\log\log
n+o(\log\log n))\\&<(\log q)(7\log\log n+o(\log\log
n)).\end{align*} Hence,
$$h_j(\bm x)<q^{7\log\log n+o(\log\log n)}=N_2.$$

Now, we prove that each $\bm{x}\in\mathcal R_n$ can be uniquely
recovered from $f(\bm x)=\left(\eta(c_{1,[n]}),g^{(0)}(\bm
x),g^{(1)}(\bm x), h^{(0)}(\bm x),h^{(1)}(\bm x)\right)$ and any
given $\bm y\in\mathcal D_2(\bm x)$. Let $$M_{\bm
y}=(d_{i,j})_{\left\lceil\log q\right\rceil\times (n-2)}$$ be the
matrix representation of $\bm y$. Then $M_{\bm y}$ can be obtained
from $M_{\bm x}$ by deleting two columns, so
$d_{1,[n-2]}\in\mathcal D_2(c_{1,[n]})$, where $d_{1,[n-2]}$ is
the first row of $M_{\bm y}$. Since $\bm x\in\mathcal R_n$, then
$c_{1,[n]}$ is regular. By Lemma \ref{Reg-Bnry-2-Del}, $\bm
c\triangleq c_{1,[n]}$ can be correctly recovered from $\bm
d\triangleq d_{1,[n-2]}$ and $\eta(c_{1,[n]})$. Moreover, by Lemma
\ref{lem-2del-pstn}, exact one of the following two cases holds:

\emph{Case 1}: There are two distinct runs $c_{1,I_{j_1}}$ and
$c_{1,I_{j_2}}$ of $c_{1,[n]}$ such that $d_{1,[n-2]}$ is obtained
from $c_{1,[n]}$ by deleting one symbol in $c_{1,I_{j_1}}$ and one
symbol $c_{1,I_{j_2}}$. Correspondingly, $M_{\bm y}$ can be
obtained from $M_{\bm x}$ by deleting one column in $I_{j_1}$ and
one column in $I_{j_2}$. Without loss of generality, assume
$j_1<j_2$. Denoting $I_{j}=[p_{j-1}+1,p_j]$, where
$p_0=0<p_1<p_2<\cdots <p_{n'}=n$, then by comparing the symbols of
$M_{\bm x}$ and $M_{\bm y}$, we have the following observation:
\begin{itemize}
 \item[i)] $c_{i,I_j}=d_{i,I_j}$ for $1\leq j<j_1$ and each
 $i\in[2,\left\lceil\log q\right\rceil]$;
 \item[ii)] $d_{i,[p_{j_1-1}+1,p_{j_1}-1]}\in \mathcal
 D_1(c_{i,I_{j_1}})$ for each
 $i\in[2,\left\lceil\log q\right\rceil]$;
 \item[iii)] $c_{i,I_j}=d_{i,I_j-1}$ for each $j_1<j<j_2$ and
 $i\in[2,\left\lceil\log q\right\rceil]$, where
 $I_j-1=\{\ell-1:\ell\in I_{j}\}=[p_{j-1},p_j-1]$;
 \item[iv)] $d_{i,[p_{j_2-1},p_{j_2}-2]}\in \mathcal
 D_1(c_{i,I_{j_2}})$ for each
 $i\in[2,\left\lceil\log q\right\rceil]$;
 \item[v)] $c_{i,I_j}=d_{i,I_j-2}$ for each $j_2<j\leq n'$ and
 $i\in[2,\left\lceil\log q\right\rceil]$, where
 $I_j-2=\{\ell-2:\ell\in I_{j}\}=[p_{j-1}-1,p_j-2]$.
\end{itemize}
Then $c_{i,[n]}$, $i\in[2,\left\lceil\log q\right\rceil]$, can be
recovered by the following three steps.

\textbf{Step 1}: By observations i), iii) and v), $c_{i,I_j}$ can
be directly obtained from $M_{\bm y}$ for all
$i\in[2,\left\lceil\log q\right\rceil]$ and
$j\in[n']\backslash\{j_1,j_2\}$.

\textbf{Step 2}: Compute $g_j(\bm x)$ from $c_{2,I_j}, \cdots,
c_{\left\lceil\log q\right\rceil,I_j}$ for all
$j\in[n']\backslash\{j_1,j_2\}$. This is possible because for all
$i\in[2,\left\lceil\log q\right\rceil]$ and
$j\in[n']\backslash\{j_1,j_2\}$, $c_{i,I_j}$ have been obtained
from $M_{\bm y}$ in Step 1. Then by taking $\ell=0$ in
\eqref{def-2del-g-ell}, we can obtain
\begin{align*}g_{j_1}(\bm x)+g_{j_2}(\bm x)\equiv \left(\!g^{(0)}(\bm
x)-\!\sum_{j\in[n']\backslash\{j_1,j_2\}}g_j(\bm
x)\!\right)\text{mod}~2N_1.\end{align*} Since $g_{j}(\bm x)<N_1$
for all $j\in[n']$, we in fact have
\begin{align}\label{slv-g1}
g_{j_1}(\bm x)+g_{j_2}(\bm x)= \left(\!g^{(0)}(\bm
x)-\!\sum_{j\in[n']\backslash\{j_1,j_2\}}g_j(\bm
x)\!\right)\text{mod}~2N_1.\end{align} Similarly, taking $\ell=1$
in \eqref{def-2del-g-ell} and noticing that $0<j_1g_{j_1}(\bm
x)+j_2g_{j_2}(\bm x)<2nN_1$, we can obtain
\begin{align}\label{slv-g2}&j_1g_{j_1}(\bm x)+j_2g_{j_2}(\bm x)\nonumber\\
&=\left(g^{(1)}(\bm
x)-\sum_{j\in[n']\backslash\{j_1,j_2\}}jg_j(\bm
x)\right)~\text{mod}~2nN_1.
\end{align}
So, $g_{j_1}(\bm x)$ and $g_{j_2}(\bm x)$ can be solved from
\eqref{slv-g1} and \eqref{slv-g2}. By Construction 1, we have
$$g_{j_1}(\bm x)=\left(\text{VT}(c_{2,I_{j_1}}),
\text{VT}(c_{3,I_{j_1}}), \cdots, \text{VT}(c_{\lceil\log
q\rceil,I_{j_1}})\right)$$ and $$g_{j_2}(\bm
x)=\left(\text{VT}(c_{2,I_{j_2}}), \text{VT}(c_{3,I_{j_2}}),
\cdots, \text{VT}(c_{\lceil\log q\rceil,I_{j_2}})\right).$$

\textbf{Step 3}: By observation ii),
$d_{i,[p_{j_1-1}+1,p_{j_1}-1]}\in \mathcal D_1(c_{i,I_{j_1}})$ for
each $i\in[2,\left\lceil\log q\right\rceil]$. Then by Lemma
\ref{VT-Code-Skch}, $c_{i,I_{j_1}}$ can be recovered from
$\text{VT}(c_{i,I_{j_1}})$ and $d_{i,[p_{j_1-1}+1,p_{j_1}-1]}$.
Similarly, since by observation iv),
$d_{i,[p_{j_2-1},p_{j_2}-2]}\in \mathcal D_1(c_{i,I_{j_2}})$ for
each $i\in[2,\left\lceil\log q\right\rceil]$, then by Lemma
\ref{VT-Code-Skch}, $c_{i,I_{j_2}}$ can be recovered from
$\text{VT}(c_{i,I_{j_2}})$ and $d_{i,[p_{j_2-1},p_{j_2}-2]}$.

Thus, for case 1, $c_{i,[n]}$, $i\in[2,\left\lceil\log
q\right\rceil]$, can be recovered from $\eta(c_{1,[n]})$,
$g^{(0)}(\bm x)$, $g^{(1)}(\bm x)$ and $\bm y$.

\emph{Case 2}: There is an interval $J\subseteq[n]$ of length at
most $\rho=3d\log n$ such that $\bm d\triangleq d_{1,[n-2]}$ can
be obtained from $\bm c\triangleq c_{1,[n]}$ by deleting two
symbols in $c_{1,J}$. Correspondingly, $M_{\bm y}$ can be obtained
from $M_{\bm x}$ by deleting two columns in $J$. By 1) of Remark
\ref{rem-sets-Ji}, we can always find an $J_{j_0}$ for some
$j_0\in\{1,2,\cdots,\left\lceil n/\rho\right\rceil-1\}$ such that
$J\subseteq J_{j_0}$. Denoting $J_{j_0}=[\lambda, \lambda']$, then
by comparing the symbols of $M_{\bm x}$ and $M_{\bm y}$, we obtain
that for each $i\in[2,\left\lceil\log q\right\rceil]$,
$$c_{i,[1,\lambda-1]}=d_{i,[1,\lambda-1]},$$ $$c_{i,[\lambda'+1,n]}
=d_{i,[\lambda'-1,n-2]}$$ and
$$d_{i,[\lambda,\lambda'-2]}\in\mathcal D_{\leq 2}(c_{i,[\lambda,\lambda']}).$$
Hence, $c_{i,[1,\lambda-1]}$ and $c_{i,[\lambda'+1,n]}$ can be
directly obtained from $M_{\bm y}$. Moreover, each
$c_{i,[\lambda,\lambda']}$ can be recovered from
$d_{i,[\lambda,\lambda'-2]}$ and $h^{(\ell)}(\bm x)$,
$\ell\in\{0,1\}$, as follows.

By 2) of Remark \ref{rem-sets-Ji}, $J_j\subseteq[1,\lambda-1]$ for
all $j\in\{1$, $2$, $\cdots$, $j_0-2\}$, so $c_{i,J_j}$ can be
obtained from $d_{i,[1,\lambda-1]}=c_{i,[1,\lambda-1]}$.
Similarly, $c_{i,J_j}$ can be obtained from
$d_{i,[\lambda'+1-t,n]}=c_{i,[\lambda'+1,n]}$ for all
$j\in\{j_0+2,\cdots,\left\lceil n/\delta'\right\rceil-1\}$. Hence,
we can compute $h_j(\bm
x)=\left(\xi(c_{2,J_j}),\xi(c_{3,J_j}),\cdots,\xi(c_{\lceil\log
q\rceil,J_j})\right)$ for each $j\in\{1,2,\cdots,\left\lceil
n/\delta'\right\rceil-1\}\backslash\{j_0\}$. Let $\ell\in\{0,1\}$
be such that $j_0\equiv\ell\mod 2$. Then by \eqref{Constr1-h-ell},
and noticing that $h_{j_0}(\bm x)<N_2$, we can obtain
\begin{align*}
h_{j_0}(\bm x)=h^{(\ell)}(\bm
x)-\sum_{\substack{j\in\{1,2,\cdots,\left\lceil
n/\delta'\right\rceil-1\}\backslash\{j_0\}:\\
j\!~\equiv\!~\ell~\text{mod}~2}} h_j(\bm
x)~\text{mod}~N_2.\end{align*} Note that
$b_{i,[\lambda,\lambda'-2]}\in\mathcal D_{\leq
2}(c_{i,[\lambda,\lambda']})=\mathcal D_{\leq 2}(c_{i,J_{j_0}})$,
and by Construction 1,
$$h_{j_0}(\bm
x)=\left(\xi(c_{2,J_{j_0}}),\xi(c_{3,J_{j_0}}),\cdots,\xi(c_{\lceil\log
q\rceil,J_{j_0}})\right).$$ Then by Lemma \ref{Ind-Bnry-2del}, for
each $i\in[2,\left\lceil\log q\right\rceil]$,
$c_{i,[\lambda,\lambda']}=c_{i,J_{j_0}}$ can be recovered from
$d_{i,[\lambda,\lambda'-2]}$ and $h_{j_0}(\bm x)$.

Thus, for case 2, $c_{i,[n]}$, $i\in[2,\left\lceil\log
q\right\rceil]$, can be recovered from $\eta(c_{1,[n]})$,
$h^{(\ell)}(\bm x)$ and $\bm y$.

By the above discussions, we proved that $M_{\bm x}~($and so $\bm
x)$ can be uniquely recovered from $f(\bm x)$ and $\bm y$, which
completes the proof.
\end{proof}

\vspace{0pt}By representing each binary string of length at most
$\lfloor\log q\rfloor$ as an integer in $\mathbb Z_q$, each binary
string $\bm a$ can be represented as a $q$-ary string of length
$\left\lceil|\bm a|/\lfloor\log q\rfloor\right\rceil$. We denote
this $q$-ary string by $\mathcal Q(\bm a)$ and call it the
$q$-\emph{ary representation} of $\bm a$ for convenience of use.
Specifically, divide $\bm a$ into $\left\lceil|\bm a|/\lfloor\log
q\rfloor\right\rceil$ disjoint substrings, each having length
$\lfloor\log q\rfloor$ except the last substring which has length
$|\bm a|-\left(\left\lceil|\bm a|/\lfloor\log
q\rfloor\right\rceil-1\right)\lfloor\log q\rfloor$. Then by
representing each of these substrings as an integer in $\mathbb
Z_q$, we can obtain a $q$-ary string $\mathcal Q(\bm a)$ of length
$\left\lceil|\bm a|/\lfloor\log q\rfloor\right\rceil$.

Let $\text{RegEnc}: \{1,2,\cdots,M\}\rightarrow\{0,1\}^n$ be the
one-to-one mapping constructed in Lemma \ref{Enc-Reg-Bnry}. Since
$M\geq 2^{n-1}$, then $\text{RegEnc}$ can also be viewed as a
mapping from $\{0,1\}^{n-1}$ to $\{0,1\}^n$. By Lemma
\ref{lem-En-B2Q}, the mapping $\text{RegEnc}$ can be extended to a
one-to-one mapping, denoted by $$\mathcal E_{\text{Reg}}:\mathbb
Z_q^{n-1}\rightarrow\mathbb Z_q^{n},$$ such that for any $\bm
u\in\mathbb Z_q^{n-1}$ and $\bm x=\mathcal E_{\text{Reg}}(\bm u)$,
if $M_{\bm{u}}=(b_{i,j})_{\lceil\log q\rceil\times(n-1)}$ and
$M_{\bm{x}}=(c_{i,j})_{\lceil\log q\rceil\times n}$ are the matrix
representation of $\bm u$ and $\bm x$ respectively, then
$$c_{1,[n]}=\text{RegEnc}(b_{1,[n-1]}).$$ By Lemma
\ref{Enc-Reg-Bnry}, $c_{1,[n]}=\text{RegEnc}(b_{1,[n-1]})$ is
regular, so for any $\bm u\in\mathbb Z_q^{n-1}$, we have $\bm
x=\mathcal E_{\text{Reg}}(\bm u)\in\mathcal R_n$.

Using the mapping $\mathcal E_{\text{Reg}}:\mathbb
Z_q^{n-1}\rightarrow\mathcal R_n$ and the function $f$ constructed
in Construction 1, we can give an encoding function of a $q$-ary
two-deletion correcting code as follows.

%\textbf{Encoding function}:
Let $\mathcal E$ be the function defined on $\mathbb Z_q^{n-1}$ of
the form
\begin{align}\label{2del-enc-fun}
\mathcal E(\bm u)=(\bm v, \bm v', \bm v''), ~\forall \!~\bm
u\in\mathbb Z_q^{n-1},
\end{align}
such that $\bm v=\mathcal E_{\text{Reg}}(\bm u)$, $\bm v'=\mathcal
E_{\text{Reg}}(\mathcal Q(f(\bm v)))$ and $\bm
v''=\text{Rep}_3(\mathcal Q(f(\bm v')))$, where
$\text{Rep}_{3}(\cdot)$ is the encoding function of the $3$-fold
repetition code.

\begin{thm}\label{thm-2del-enc}
The code $\mathcal C=\{\mathcal E(\bm u): \bm u\in\mathbb
Z_q^{n-1}\}$, where $\mathcal E$ is given by \eqref{2del-enc-fun},
is a $q$-ary two-deletion correcting code with redundancy $5\log
n+O(\log q\log\log n)$ in bits. The encoding complexity of
$\mathcal C$ is near-linear in $n$ with a polynomial size lookup
table.
\end{thm}
\begin{proof}
Let $$\bm x=\mathcal E(\bm u)=(\bm v, \bm v', \bm v'')\in\mathcal
C,$$ where $\bm u\in\mathbb Z_q^{n-1}$, $\bm v=\mathcal
E_{\text{Reg}}(\bm u)$, $\bm v'=\mathcal Q(f(\bm v))$ and $\bm
v''=\text{Rep}_3(\mathcal Q(f(\bm v')))$ as in
\eqref{2del-enc-fun}. Given any $\bm y\in\mathcal D_2(\bm x)$, we
have $y_{[1,m_1-2]}\in\mathcal D_2(\bm v)$,
$y_{[m_1,m_2-2]}\in\mathcal D_2(\bm v')$ and
$y_{[m_2,m_3-2]}\in\mathcal D_2(\bm v'')$, where $|\bm v|=m_1,
|(\bm v,\bm v')|=m_2$ and $|\bm x|=|(\bm v,\bm v',\bm v'')|=m_3$.
First, since $\bm v''=\text{Rep}_3(\mathcal Q(f(\bm v')))$ is a
codeword of a two-deletion code, then $\mathcal Q(f(\bm v'))$ can
be recovered from $y_{[m_2,m_3-2]}$, and hence $f(\bm v')$ can be
recovered from $\mathcal Q(f(\bm v'))$. Then by Theorem
\ref{thm-2del-sketch}, $\bm v'$ can be recovered from
$y_{[m_1,m_2-2]}$ and $f(\bm v')$, and hence $f(\bm v)$ can be
recovered from $\bm v'=\mathcal E_{\text{Reg}}(\mathcal Q(f(\bm
v)))$. Finally, by Theorem \ref{thm-2del-sketch} again, $\bm v$
can be recovered from $y_{[1,m_1-2]}$ and $f(\bm v)$. Thus, $\bm
x=(\bm v, \bm v', \bm v'')$ can be recovered from any $\bm
y\in\mathcal D_2(\bm x)$, which proves that $\mathcal C$ is a
two-deletion correcting code.

Since $\bm u\in\mathbb Z_q^{n-1}$ and $\bm v=\mathcal
E_{\text{Reg}}(\bm u)\in\mathbb Z_q^{n}$, so $\bm v$ has $\log q$
bits redundancy. Moreover, by Theorem \ref{thm-2del-sketch}, the
length of $\bm v'$ is $$|\bm v'|=5\log n+O(\log q\log\log n)$$
bits and the length of $\bm v''$ is $$|\bm v''|=3(5\log |\bm
v'|+O(\log q\log\log |\bm v'|))=O(\log\log n)$$ bits. So the total
redundancy of $\bm x=\mathcal E(\bm u)$ is
\begin{align*}\text{redundancy~of}~\bar{\mathcal
C}&=\log q+|\bm v'|+|\bm v''|\\&=5\log n+O(\log q\log\log
n)\end{align*} in bits.

By Lemma \ref{Enc-Reg-Bnry} and Lemma \ref{lem-En-B2Q}, the
encoding complexity of $\bm v=\mathcal E_{\text{Reg}}(\bm u)$ is
near-linear in $n$ with a polynomial size lookup table. Moreover,
by Theorem \ref{thm-2del-sketch}, the encoding complexity of $\bm
v'=\mathcal E_{\text{Reg}}(\mathcal Q(f(\bm v)))$ and $\bm
v''=\text{Rep}_3(\mathcal Q(f(\bm v')))$ is linear in $n$ and
$\log n$ respectively. Therefore, the encoding complexity of
$\mathcal E(\bm u)=(\bm v, \bm v', \bm v'')$ is near-linear in $n$
with a polynomial size lookup table, which completes the proof.
\end{proof}

\section{Binary Codes Correcting a Burst of at most $t$ Deletions}

In this section, we present a construction of binary codes that
are capable of correcting a bursting of at most $t$ deletions
improving the Lenz-Polyanskii Construction in \cite{Lenz20}. We
assume that $t$ is a constant with respect to the code length $n$,
and for notational simplicity, we use $\gamma_t$ to denote any
constant that depends only on $t$. As in Section III, each binary
sequence $\bm a$ is identified with the positive integer whose
binary representation is $\bm a$, and summation and multiplication
of binary strings and integers are performed in the set of
integers.

Recall that a string $\bm c\in\{0,1\}^n$ is called $(\bm p,
\delta)$-dense, if each substring of $\bm c$ of length $\delta$
contains at least one pattern $\bm p$. As stated in Section II, we
take
$$\delta=t2^{t+1}\log n$$ and
$$\bm p=0^t1^t.$$

The basic idea of our construction is to replace the shifted VT
code in the Lenz-Polyanskii Construction with the function $\phi$
constructed in Lemma \ref{lem-Bnry-burst-Sima}. To apply the
function $\phi$, we need to divide each binary string into
substrings of length at most $2(\delta+t)$. Specifically, we
denote $\delta'=\delta+t$ and let
\begin{equation}\label{def-Li-intvl}
L_i=\!\left\{\!\begin{aligned} &[(i-1)\delta'\!+1, (i+1)\delta'],
~\text{for}~i\in\!\{1,\cdots, \left\lceil n/\delta'\right\rceil-2\},\\
&[(i-1)\delta'\!+1, n], ~~~~~~~\!~~~\text{for}~i=\left\lceil
n/\delta'\right\rceil-1,\vspace{12pt}
\end{aligned}\right.
\end{equation}
where $i\in\{1,\cdots, \left\lceil n/\delta'\right\rceil-1\}$, be
the index sets of the expected substrings. Then we can construct a
function $\bar{f}^{\text{b}}$ over $\{0,1\}^n$ as follows, which
is the main component of our construction of binary burst-deletion
correcting codes.

\textbf{Construction 2}: Let $\phi$ and $\mu$ be the functions
constructed in Lemma \ref{lem-Bnry-burst-Sima} and Lemma
\ref{lem-Bnry-burst-Lenz} respectively. For each $\bm
c\in\{0,1\}^n$, let
\begin{align*}%\label{Constr2-f-b}
\bar{f}^{\text{b}}(\bm c)=\left(\mu(\bm c), \bar{g}^{(0)}(\bm c),
\bar{g}^{(1)}(\bm c)\right),\end{align*} such that for each
$\ell\in\{0,1\}$,
\begin{align}\label{Constr2-g-b}
\bar{g}^{(\ell)}(\bm
c)=\sum_{\substack{i\in\{1,2,\cdots,\left\lceil
n/\lceil\delta'\rceil\right\rceil-1\}:\\
i~\equiv~ \ell~\text{mod}~2}}
\phi(c_{L_i})~\text{mod}~\overline{N}^{\text{b}},\end{align} where
\begin{align*}\overline{N}^{\text{b}}\triangleq
2^{4\log(2\delta')+o(\log(2\delta'))}=2^{4\log\log
n+\gamma_t+o(\log\log n)}.\footnotemark{}\end{align*}
\footnotetext{Since $\delta'=\delta+t=t2^{t+1}(\log n+2^{-t-1})$,
so more accurately, it should be
$\overline{N}^{\text{b}}\triangleq
2^{4\log(2\delta')+o(\log(2\delta'))}=2^{4\log(\log
n+2^{-t-1})+\gamma_t+o(\log\log n)}$. However, because
$\overline{N}^{\text{b}}$ is an integer, so for sufficiently large
$n$, we can always obtain $\overline{N}^{\text{b}}\triangleq
2^{4\log(2\delta')+o(\log(2\delta'))}=2^{4\log\log
n+\gamma_t+o(\log\log n)}$.}

For Construction 2, we have the following theorem.
\begin{thm}\label{thm-b-burst-del-sketch}
For each $\bm c\in\{0,1\}^n$, $\bar{f}^{\text{b}}(\bm c)$ is
computable in linear time and the length $|\bar{f}^{\text{b}}(\bm
c)|$ of $\bar{f}^{\text{b}}(\bm c)$ satisfies
$$|\bar{f}^{\text{b}}(\bm c)|\leq\log n+8\log\log
n+\gamma_t+o(\log\log n).$$ Moreover, if $\bm c$ is $(\bm p,
\delta)$-dense, then given $\bar{f}^{\text{b}}(\bm c)$ and any
$\bm b\in\mathcal B_{\leq t}(\bm c)$, one can uniquely recover
$\bm c$.
\end{thm}

Before proving Theorem \ref{thm-b-burst-del-sketch}, we give some
remark on the properties of the sets $L_j,
j=1,2,\cdots,\left\lceil n/\delta'\right\rceil-1$.
\begin{rem}\label{rem-sets-Li}
Similar to Remark \ref{rem-sets-Ji}, it is easy to see that
\begin{itemize}
 \item[1)] For each interval $L\subseteq[n]$ of length at most
 $\delta'=\delta+t$, we can always find an
 $i_0\in\{1,2,\cdots,\left\lceil n/\delta'\right\rceil-1\}$
 such that $L\subseteq L_{i_0}$.
 \item[2)] $L_i\cap L_{i'}=\emptyset$ for all $i,i'\in\{1,2,\cdots,
 \left\lceil n/\delta'\right\rceil-1\}$ such that $|i-i'|\geq 2$.
\end{itemize}
\end{rem}

Now, we can prove Theorem \ref{thm-b-burst-del-sketch}.
\begin{proof}
Note that by Lemma \ref{lem-Bnry-burst-Lenz}, $\mu(\bm c)$ is
computable in linear time. By Lemma \ref{lem-Bnry-burst-Sima},
each $\phi(c_{L_i})$ is computable in time
$O(2^t(2\delta)^3)=O((\log n)^3)$, so $(\bar{g}^{(0)}(\bm
c),\bar{g}^{(1)}(\bm c))$ are also computable in linear time.
Hence, by Construction 2, $\bar{f}^{\text{b}}(\bm c)=\left(\mu(\bm
c), \bar{g}^{(0)}(\bm c), \bar{g}^{(1)}(\bm c)\right)$ is
computable in linear time. Moreover, by Lemma
\ref{lem-Bnry-burst-Lenz} and \eqref{Constr2-g-b}, the length
$|\bar{f}^{\text{b}}(\bm c)|$ of $\bar{f}^{\text{b}}(\bm c)$
satisfies
\begin{align*}|\bar{f}^{\text{b}}(\bm c)|&=|\mu(\bm c)|+
|\bar{g}^{(0)}(\bm c)|+|\bar{g}^{(1)}(\bm c)|\\&\leq\log
n+3+2\big(4\log\log n+\gamma_t+o(\log\log n)\big)\\&=\log
n+8\log\log n+\gamma_t+o(\log\log n).\end{align*}

Suppose $\bm c$ is $(\bm p, \delta)$-dense and $\bm b\in\mathcal
B_{\leq t}(\bm c)$. We need to prove that $\bm c$ can be uniquely
recovered from $\bm b$ and $\bar{f}^{\text{b}}(\bm c)$.

By Lemma \ref{lem-Bnry-burst-Lenz}, we can find an interval
$L\subseteq[n]$ of length at most $\delta'=\delta+t$ such that
$\bm b=c_{[n]\backslash D}$ for some interval $D\subseteq L$ of
length $t'=|\bm c|-|\bm b|$. By 1) of Remark \ref{rem-sets-Li}, we
can always find an $i_0\in\{1,2,\cdots,\left\lceil
n/\delta'\right\rceil-1\}$ such that $L\subseteq L_{i_0}$.
Denoting $L_{i_0}=[\lambda, \lambda']$, then we can obtain
$$c_{[1,\lambda-1]}=b_{[1,\lambda-1]},$$ $$c_{[\lambda'+1,n]}
=b_{[\lambda'-t'+1,n]}$$ and
$$b_{[\lambda,\lambda'-t']}\in\mathcal B_{\leq t}(c_{[\lambda,\lambda']}).$$
Therefore, $c_{[1,\lambda-1]}$ and $c_{[\lambda'+1,n]}$ can be
directly obtained from $\bm b$. In the following, we will show how
to recover $c_{[\lambda,\lambda']}$ from
$b_{[\lambda,\lambda'-t']}$ and $\bar{g}^{(\ell)}(\bm c)$ for some
$\ell\in\{0,1\}$.

By 2) of Remark \ref{rem-sets-Li}, for all
$i\in\{1,2,\cdots,i_0-2\}$, we have $L_i\subseteq[1,\lambda-1]$,
so $c_{L_i}$ can be obtained from $b_{[1,\lambda-1]}$ and hence
$\phi(c_{L_i})$ can be computed. Similarly, for all
$i\in\{i_0+2,\cdots,\left\lceil n/\delta'\right\rceil-1\}$,
$c_{L_i}$ can be obtained from $b_{[\lambda'-t'+1,n]}$ and hence
$\phi(c_{L_i})$ can be computed. Let $\ell_0\in\{0,1\}$ be such
that $\ell_0~\equiv~ i_0~\text{mod}~2$. By \eqref{Constr2-g-b}, we
have
\begin{align*}
\phi(c_{L_{i_0}})\equiv\bar{g}^{(\ell_0)}(\bm
c)-\sum_{\substack{i\in\{1,\cdots,\left\lceil
n/\lceil\delta'\rceil\right\rceil-1\}:\\
i\!~\neq\!~ i_0\!~\text{and}\!~ i\!~\equiv\!~
\ell_0\!~\text{mod}\!~2}}
\phi(c_{L_i})~\text{mod}~\overline{N}^{\text{b}}.\end{align*} By
\eqref{def-Li-intvl}, $|L_{i_0}|=2\delta'=2(\delta+t)$, so by
Lemma \ref{lem-Bnry-burst-Sima}, $\phi(c_{L_{i_0}})\leq
2^{4\log(2\delta')+o(\log(2\delta'))}=\overline{N}^{\text{b}}$.
Therefore, we actually have
\begin{align*}
\phi(c_{L_{i_0}})=\bar{g}^{(\ell_0)}(\bm
c)-\sum_{\substack{i\in\{1,\cdots,\left\lceil
n/\lceil\delta'\rceil\right\rceil-1\}:\\
i\!~\neq\!~ i_0\!~\text{and}\!~ i\!~\equiv\!~
\ell_0\!~\text{mod}\!~2}}
\phi(c_{L_i})~\text{mod}~\overline{N}^{\text{b}}.\end{align*}
Since $b_{[\lambda,\lambda'-t']}\in\mathcal B_{\leq
t}(c_{[\lambda,\lambda']})$, again by Lemma
\ref{lem-Bnry-burst-Sima}, we can recover $c_{[\lambda,\lambda']}$
from $b_{[\lambda,\lambda'-t]}$ and $\phi(c_{L_{i_0}})$. Note that
we have obtained $c_{[1,\lambda-1]}=b_{[1,\lambda-1]}$ and
$c_{[\lambda'+1,n]} =b_{[\lambda'-t+1,n]}$, so $\bm c$ can be
uniquely recovered, which completes the proof.
\end{proof}

Let $\mathcal S^{\text{b}}_n$ be the set of all $(\bm p,
\delta)$-dense binary strings $\bm c\in\{0,1\}^{n}$, where $\bm
p=0^t1^t$ and $\delta=t2^{t+1}\lceil\log n\rceil$. By Lemma
\ref{lem-p-dense}, there is a one-to-one mapping that maps each
binary string of length $n-1$ to a string in $\mathcal
S^{\text{b}}_n$. For convenience, we denote this mapping by
\begin{align}\label{E-B-Den}\mathcal
E^{\text{b}}_{\text{Den}}:\{0,1\}^{n-1}\rightarrow\mathcal
S^{\text{b}}_n.\end{align} Using the function $\bar{f}^{\text{b}}$
constructed in Construction 2, we can construct an encoding
function of a binary code capable of correcting a burst of at most
$t$ deletions.

Let $\bar{\mathcal E}^{\text{b}}$ be a function defined on
$\{0,1\}^{n-1}$ of the form
\begin{align}\label{BBtdel-enc-fun}
\bar{\mathcal E}^{\text{b}}(\bm a)=(\bm b, \bm b', \bm b''),
~~\forall\!~\bm a\in\{0,1\}^{n-1},
\end{align}
such that $\bm b=\mathcal E^{\text{b}}_{\text{Den}}(\bm a)$, $\bm
b'=\mathcal E^{\text{b}}_{\text{Den}}(\bar{f}^{\text{b}}(\bm b))$
and $\bm b''=\text{Rep}_{t+1}(\bar{f}^{\text{b}}(\bm b'))$, where
$\text{Rep}_{t+1}(\cdot)$ is the encoding function of the
$(t+1)$-fold repetition code.

\begin{thm}\label{thm-bbdel-enc}
The code $\bar{\mathcal C}^{\text{b}}=\{\bar{\mathcal
E}^{\text{b}}(\bm a): \bm a\in\{0,1\}^{n-1}\}$, where
$\bar{\mathcal E}^{\text{b}}$ is given by \eqref{BBtdel-enc-fun},
is a binary code with redundancy $\log n+9\log\log
n+\gamma_t+o(\log\log n)$ bits and capable of correcting a burst
of at
most $t$ deletions. %The encoding complexity of $\mathcal C$ is
%near-linear in $n$ with a polynomial size lookup table and the
%decoding complexity of $\mathcal C$ is linear in $n$.
\end{thm}
\begin{proof}
Let $$\bm c=\bar{\mathcal E}^{\text{b}}(\bm a)=(\bm b, \bm b', \bm
b'')\in\bar{\mathcal C}^{\text{b}},$$ where $\bm
a\in\{0,1\}^{n-1}$, $\bm b=\mathcal E^{\text{b}}_{\text{Den}}(\bm
a)$, $\bm b'=\mathcal
E^{\text{b}}_{\text{Den}}(\bar{f}^{\text{b}}(\bm b))$ and $\bm
b''=\text{Rep}_{t+1}(\bar{f}^{\text{b}}(\bm b'))$. Given any $\bm
d\in\mathcal B_{\leq t}(\bm c)$, denoting $t'=|\bm c|-|\bm b|$,
then $t'\leq t$ and we have $d_{[1,m_1-t']}\in\mathcal B_{\leq
t}(\bm b)$, $d_{[m_1,m_2-t']}\in\mathcal B_{\leq t}(\bm b')$ and
$d_{[m_2,m_3-t']}\in\mathcal B_{\leq t}(\bm b'')$, where $m_1=|\bm
b|, m_2=|(\bm b,\bm b')|$ and $m_3=|\bm c|=|(\bm b,\bm b',\bm
b'')|$. First, since $\bm
b''=\text{Rep}_{t+1}(\bar{f}^{\text{b}}(\bm b'))$ is a codeword of
a $t$-deletion code, then $\bar{f}^{\text{b}}(\bm b')$ can be
recovered from $d_{[m_2,m_3-t']}$. Further, by Theorem
\ref{thm-b-burst-del-sketch}, $\bm b'$ can be recovered from
$d_{[m_1,m_2-t']}$ and $\bar{f}^{\text{b}}(\bm b')$, and so
$\bar{f}^{\text{b}}(\bm b)$ can be recovered from $\bm b'=\mathcal
E^{\text{b}}_{\text{Den}}(\bar{f}^{\text{b}}(\bm b))$. Finally, by
Theorem \ref{thm-b-burst-del-sketch} again, $\bm b$ can be
recovered from $d_{[1,m_1-t']}$ and $\bar{f}^{\text{b}}(\bm b)$.
Thus, $\bm c=(\bm b, \bm b', \bm b'')$ can be recovered from any
$\bm d\in\mathcal B_{\leq t}(\bm c)$, which proves that
$\bar{\mathcal C}^{\text{b}}$ is capable of correcting a burst of
at most $t$ deletions.

Since $\bm a\in\{0,1\}^{n-1}$ and $\bm b=\mathcal
E^{\text{b}}_{\text{Den}}(\bm a)\in\mathcal
S^{\text{b}}_n\subseteq\{0,1\}^{n}$, so $\bm b$ has one bit
redundancy. Moreover, by Theorem \ref{thm-b-burst-del-sketch}, the
length of $\bm b'$ is $$|\bm b'|=\log n+8\log\log
n+\gamma_t+o(\log\log n)$$ bits and the length of $\bm b''$ is
\begin{align*}
|\bm b''|&=\log |\bm b'|+8\log\log |\bm b'|+\gamma_t+o(\log\log
|\bm b'|)\\&=\log\log n+\gamma_t+o(\log\log n)\end{align*} bits.
So the total redundancy of $\bm c=\bar{\mathcal E}^{\text{b}}(\bm
a)$ is
\begin{align*}\text{redundancy~of}~\bar{\mathcal
C}&=1+|\bm b'|+|\bm b''|\\&=\log n+9\log\log n+\gamma_t+o(\log\log
n)\end{align*} bits.
\end{proof}

\section{$q$-ary Codes Correcting a Burst of at most $t$ Deletions}

In this section, we construct $q$-ary codes correcting a bursting
of at most $t$ deletions, where $q>2$ is an even integer. We
assume that $q$ and $t$ are constant with respect to the code
length $n$. As in Section III, we identify each binary string $\bm
a$ with the positive integer whose binary representation is $\bm
a$. As stated in Section II, we take
$$\delta=t2^{t+1}\log n$$ and
$$\bm p=0^t1^t.$$ A string $\bm c\in\{0,1\}^n$ is called $(\bm p,
\delta)$-dense, if each substring of $\bm c$ of length $\delta$
contains at least one pattern $\bm p$.

For each $\bm x\in\mathbb Z_q^n$, let
$M_{\bm{x}}=(c_{i,j})_{\lceil\log q\rceil\times n}$ be the matrix
representation of $\bm x$ as defined by \eqref{q-B-Repr}. Then for
each $t'\in[t]$, the deletion of $x_{i},x_{i+1},\cdots,x_{i+t'-1}$
results in the deletion of the columns $i,i+1,\cdots,i+t'-1$ of
$M_{\bm{x}}$. A basic idea is to protect the first row $\bm
c=c_{1,[n]}$ by a burst-deletion correcting code. However, in
general, if $\bm c$ can be recovered from a $\bm d\in\mathcal
B_{\leq t}(\bm c)$, the location of the deleted symbols can not be
determined. For example, consider $\bm c=0111011011010010$ and
$\bm d=0111011010010$. Then $\bm d$ can be obtained from $\bm c$
by deleting $c_3c_4c_5=110$, or deleting $c_4c_5c_6=101$. In fact,
$\bm d$ can be obtained from $\bm c$ by deleting
$c_ic_{i+1}c_{i+2}$ for all $i\in[3,10]$. To proceed, we need to
consider period of binary strings.

Let $\ell$ and $m$ be two positive integers such that $\ell\leq
m$. A string $\bm a\in\{0,1\}^m$ is said to have \emph{period}
$\ell~($or $\bm a$ is called a period-$\ell$ string$)$ if
$a_{i+\ell}=a_i$ for all $i\in[m-\ell]=\{1,2,\cdots,m-\ell\}$.
Clearly, a run of $\bm c$ of length $m$ has period $\ell$ for any
$\ell\in[m]$; a period-$2$ substring of $\bm c$ is either a run of
$\bm c$ or an alternative substring of $\bm c$.

\begin{lem}\label{lem-prd-idx-cap}
Suppose $\bm c\in\{0,1\}^n$ is $(\bm p, \delta)$-dense. Given any
$\bm d\in\mathcal B_{\leq t}(\bm c)$, it is possible to find an
interval $K\subseteq[n]$ of length at most $\delta-1$ such that if
$\bm d=c_{[n]\backslash D}$ and $D\subseteq[n]$ is an interval,
then it always holds that $D\subseteq K$.
\end{lem}
\begin{proof}
Since $\bm d\in\mathcal B_{\leq t}(\bm c)$, there is an interval
$D'\subseteq[n]$ such that $\bm d=c_{[n]\backslash D'}$. Let
$K\subseteq[n]$ be the interval such that $c_K$ is the maximal
substring of $\bm c$ satisfying: 1) $c_K$ has period $t'=|\bm
c|-|\bm d|$; 2) $c_K$ contains $c_{D'}$. We will prove that
$D\subseteq K$ for any interval $D\subseteq[n]$ such that $\bm
d=c_{[n]\backslash D}$.

Suppose $D=[i_1, i_1+t'-1]$ and $D'=[i_2, i_2+t'-1]$. Without loss
of generality, assume $i_1\leq i_2$. Since $c_{[n]\backslash
D}=\bm d=c_{[n]\backslash D'}$, we have
\begin{align*}
~&c_1~\cdots
~c_{i_1-1}~c_{i_1+t'}~c_{i_1+t'+1}\!~\cdots~c_{i_2+t'-1}
~c_{i_2+t'}~\cdots ~c_n\\
=~&c_1~\cdots ~c_{i_1-1}~~c_{i_1}~~~~c_{i_1+1}~~~\!~\cdots
~~c_{i_2-1}~~~c_{i_2+t'}~\dots~c_n.
\end{align*}
By comparing the symbols of $c_{[n]\backslash D'}$ and
$c_{[n]\backslash D''}$ in each position, we can obtain
$c_i=c_{i+t'}$ for each $i\in[i_1,i_2-1]$. So,
$c_{[i_1,i_2+t'-1]}$ is a substring of $\bm c$ of period $t'$ and
contains both $c_{D}$ and $c_{D'}$. As $c_K$ is the maximal
substring of $\bm c$ of period $t'$ that contains $c_{D'}$, so
$c_{[i_1,i_2+t'-1]}$ is contained in $c_K$. Thus, $c_{D}$ is
contained in $c_K$, which implies that $D\subseteq K$.

Since $\bm c$ is $(\bm p, \delta)$-dense, where $\bm p=0^t1^t$,
then each substring of $\bm c$ of length $\delta$ contains at
least one pattern $\bm p$. Note that for each $t'\in[t]$, we have
$p_{t}=0\neq 1=p_{t+t'}$, so each substring of $\bm c$ of length
$\delta$ can not has period $t'$. In other words, the length of
any period-$t'$ substring of $\bm c$ is at most $\delta-1$. Thus,
the length of $c_I~($and the length of $I)$ is at most $\delta-1$.
\end{proof}

Let
\begin{equation}\label{def-Ki-intvl}
K_j=\!\left\{\!\begin{aligned} &[(j-1)\delta\!+1, (j+1)\delta],
~\text{for}~j\in\!\{1,\cdots, \left\lceil n/\delta\right\rceil-2\},\\
&[(j-1)\delta\!+1, n], ~~~~~~~\!~~\!~\text{for}~j=\left\lceil
n/\delta\right\rceil-1.\vspace{12pt}
\end{aligned}\right.
\end{equation}

\begin{rem}\label{rem-sets-Ki}
Similar to Remark \ref{rem-sets-Ji}, it is easy to see that
\begin{itemize}
 \item[1)] For any interval $K\subseteq[n]$ of length at most
 $\delta$, there is an
 $j_0\in\{1,2,\cdots,\left\lceil n/\delta\right\rceil-1\}$
 such that $K\subseteq K_{j_0}$.
 \item[2)] $K_j\cap K_{j'}=\emptyset$ for all $j,j'\in\{1,2,\cdots,
 \left\lceil n/\delta\right\rceil-1\}$ such that $|j-j'|\geq 2$.
\end{itemize}
\end{rem}

Let $\phi$ be the function constructed by Lemma
\ref{lem-Bnry-burst-Sima} and $\bar{f}^{\text{b}}$ be the function
constructed in Construction 2. For each $\bm x\in\mathbb Z_q^n$,
let $M_{\bm{x}}=(c_{i,j})_{\lceil\log q\rceil\times n}$ be the
matrix representation of $\bm x$ as defined by \eqref{q-B-Repr}.
We have the following construction.

\textbf{Construction 3}: For each $\bm x\in\mathbb Z_q^n$ and each
$j\in\{1,2,\cdots, \left\lceil n/\delta\right\rceil-1\}$, let
$$\bar{h}_j(\bm x)=\left(\phi(c_{2,K_j}),\phi(c_{3,K_j}),\cdots,\phi(c_{\lceil\log
q\rceil,K_j})\right)$$ and for each $\ell\in\{0,1\}$, let
\begin{align}\label{Con3-h-ell}
\bar{h}^{(\ell)}(\bm x)=\sum_{\substack{j\in\{1,2,\cdots,
\left\lceil n/\delta\right\rceil-1\}:\\
j\!~\equiv\!~\ell~\text{mod}~2}} \bar{h}_j(\bm
x)~\text{mod}~\overline{N},\end{align} where
$$\overline{N}=q^{4\log\log n+o(\log\log n)+\gamma_t}.$$
Finally, let
\begin{align}\label{Constr2-f}
\bar{f}(\bm x)=\left(\bar{f}^{\text{b}}(c_{1,[n]}),
\bar{h}^{(0)}(\bm x), \bar{h}^{(1)}(\bm x)\right).\end{align}

We have the following theorem.

\begin{thm}\label{thm-burst-del-sketch}
For any $\bm x\in\mathbb Z_q^n$, $\bar{f}(\bm x)$ is computable in
linear time, and when viewed as a binary string, the length
$|\bar{f}(\bm x)|$ of $\bar{f}(\bm x)$ satisfies
$$|\bar{f}(\bm x)|\leq\log n+8(\log q+1)\log\log
n+o(\log q\log\log n)+\gamma_t,$$ where $\gamma_t$ is a constant
depending only on $t$. Moreover, if $\bm c=c_{1,[n]}$ is $(\bm p,
\delta)$-dense, then given $\bar{f}(\bm x)$ and any $\bm
y\in\mathcal B_{\leq t}(\bm x)$, one can uniquely recover $\bm x$.
\end{thm}
\begin{proof}
Note that by Theorem \ref{thm-b-burst-del-sketch},
$\bar{f}^{\text{b}}(c_{1,[n]})$ is computable in linear time.
Moreover, by Lemma \ref{lem-Bnry-burst-Sima} and
\eqref{def-Ki-intvl}, each $\phi(c_{2,K_j})$ is computable in time
$O(2^t(2\delta)^3)=O((\log n)^3)$, so by Construction 3,
$\bar{h}^{(\ell)}(\bm x), \ell=1,2,$ are computable in linear
time. Hence, $\bar{f}(\bm x)=\left(\bar{f}^{\text{b}}(c_{1,[n]}),
\bar{h}^{(0)}(\bm x), \bar{h}^{(1)}(\bm x)\right)$ is computable
in linear time.

By Theorem \ref{thm-b-burst-del-sketch}, the length of
$\bar{f}^{\text{b}}(c_{1,[n]})$ satisfies
$$|\bar{f}^{\text{b}}(c_{1,[n]})|\leq\log n+8\log\log
n+\gamma_t+o(\log\log n).$$ Moreover, by Construction 3, the
length of $\bar{h}^{(\ell)}(\bm x), \ell=1,2,$ satisfy
$$|\bar{h}^{(\ell)}(\bm x)|\leq\log\overline{N}
=\log q(4\log\log n+o(\log\log n)+\gamma_t).$$ Hence, the length
of $\bar{f}(\bm x)$ satisfies
\begin{align*}
|\bar{f}(\bm
x)|&=|\bar{f}^{\text{b}}(c_{1,[n]})|+|\bar{g}^{(0)}(\bm
x)|+|\bar{g}^{(1)}(\bm x)|\\&\leq\log n+8(\log q+1)\log\log
n+o(\log q\log\log n)\\&~~~+\gamma_t.\end{align*}

It remains to prove that if $\bm c=c_{1,[n]}$ is $(\bm p,
\delta)$-dense, then given $\bar{f}(\bm x)$ and any $\bm
y\in\mathcal B_{\leq t}(\bm x)$, one can uniquely recover $\bm x$.
To prove this, we first prove that $$\bar{h}_j(\bm
x)<\overline{N}$$ for each $j\in\{1,2,\cdots, \left\lceil
n/\delta\right\rceil-1\}$. In fact, by \eqref{def-Ki-intvl}, each
$c_{i,K_j}$, $i\in[2,\lceil\log q\rceil]$, has length
$2\delta=2t2^{t+1}\log n$, so by Lemma \ref{lem-Bnry-burst-Sima},
$\phi(c_{2,K_j})$ has length
$4\log(2\delta)+o(\log(2\delta))=4\log\log n+\gamma_t+o(\log\log
n)$. Hence, by Construction 3, we have
\begin{align*}
|\bar{h}_j(\bm
x)|&=|\left(\phi(c_{2,K_j}),\phi(c_{3,K_j}),\cdots,\phi(c_{\lceil\log
q\rceil,K_j})\right)|\\&=(\lceil\log q\rceil-1)\left(4\log\log
n+\gamma_t+o(\log\log n)\right)\\&<\log q\left(4\log\log
n+\gamma_t+o(\log\log n)\right),
\end{align*} which implies that $\bar{h}_j(\bm
x)<q^{4\log\log n+\gamma_t+o(\log\log n)}=\overline{N}$.

Now, we prove that $\bm x$ can be uniquely recovered from
$\bar{f}(\bm x)$ and any given $\bm y\in\mathcal B_{\leq t}(\bm
x)$, provided that $\bm c=c_{1,[n]}$ is $(\bm p, \delta)$-dense.
Let $$M_{\bm y}=(d_{i,j})_{\left\lceil\log q\right\rceil\times
(n-t')}$$ be the matrix representation of $\bm y$. Since $\bm
y\in\mathcal B_{\leq t}(\bm x)$ can be obtained from $\bm x$ by
deleting $t'$ consecutive symbols from $\bm x$, where $t'=n-|\bm
y|$ and $t'\in[t]$, then $M_{\bm y}$ can be obtained from $M_{\bm
x}$ by deleting $t'$ consecutive columns of $M_{\bm x}$. The
process of recovering $\bm x$ from $\bar{f}(\bm
x)=\left(\bar{f}^{\text{b}}(c_{1,[n]}), \bar{g}^{(0)}(\bm x),
\bar{g}^{(1)}(\bm x)\right)$ and $\bm y$ consists of the following
three steps.

\textbf{Step 1}: Since $\bm c=c_{1,[n]}$ is $(\bm p,
\delta)$-dense, then by Theorem \ref{thm-b-burst-del-sketch},
$c_{1,[n]}$ can be recovered from $d_{1,[n-t']}$ and
$\bar{f}^{\text{b}}(c_{1,[n]})$.

\textbf{Step 2}: According to Lemma \ref{lem-prd-idx-cap}, there
is an interval $K\subseteq[n]$ of length at most $\delta-1$ such
that $d_{1,[n-t']}$ is obtained from $\bm c=c_{1,[n]}$ by deleting
$t'$ consecutive symbols in $K$. Correspondingly, $M_{\bm y}$ is
obtained from $M_{\bm x}$ by deleting $t'$ consecutive columns in
$K$. By 1) of Remark \ref{rem-sets-Ki}, there is an
$j_0\in\{1,2,\cdots,\left\lceil n/\delta\right\rceil-1\}$ such
that $K\subseteq K_{j_0}$. Denote $K_{j_0}=[\lambda,\lambda']$.
Then we have the following observations:
\begin{itemize}
 \item[i)] $c_{i,[1,\lambda-1]}=d_{i,[1,\lambda-1]}$ for each
 $i\in[2,\lceil\log q\rceil]$.
 \item[ii)] $d_{i,[\lambda,\lambda'-t']}\in
 \mathcal B_{\leq t}(c_{i,[\lambda,\lambda']}$ for each
 $i\in[2,\lceil\log q\rceil]$.
 \item[iii)] $c_{i,[\lambda'+1,n]}=d_{i,[\lambda'-t'+1,n-t']}$ for each
 $i\in[2,\lceil\log q\rceil]$.
\end{itemize}
By observations i) and iii), for each $i\in[2,\lceil\log
q\rceil]$, $c_{i,[1,\lambda-1]}$ and $c_{i,[\lambda'+1,n]}$ can be
directly obtained from $M_{\bm y}$.

\textbf{Step 3}: For $j\in\{1,\cdots, j_0-2\}$, we have $K_j\cap
K_{j_0}=\emptyset$, by 2) of Remark \ref{rem-sets-Ki}, so
$K_j\subseteq[1,\lambda-1]$ and by observation i), $\bar{h}_j(\bm
x)=\left(\phi(c_{2,K_j}),\phi(c_{3,K_j}),\cdots,\phi(c_{\lceil\log
q\rceil,K_j})\right)$ can be computed from
$d_{i,[1,\lambda-1]}=c_{i,[1,\lambda-1]}$, $i=2,\cdots,\lceil\log
q\rceil$. Similarly, for $j\in\{j_0+2,\cdots,\left\lceil
n/\delta\right\rceil-1\}$, by 2) of Remark \ref{rem-sets-Ki}, we
have $K_j\cap K_{j_0}=\emptyset$, so $K_j\subseteq[\lambda'+1,n]$
and by observation iii), $\bar{h}_j(\bm
x)=\left(\phi(c_{2,K_j}),\phi(c_{3,K_j}),\cdots,\phi(c_{\lceil\log
q\rceil,K_j})\right)$ can be computed from
$d_{i,[\lambda'-t'+1,n-t']}=c_{i,[\lambda'+1,n]}$,
$i=2,\cdots,\lceil\log q\rceil$. Let $\ell_0\in\{0,1\}$ be such
that $\ell_0\equiv j_0~\text{mod}~2$. By \eqref{Con3-h-ell}, we
can obtain
\begin{align*}\bar{h}_{j_0}(\bm
x)\equiv\bar{h}^{(\ell_0)}(\bm
x)-\sum_{\substack{j\in\{1,2,\cdots,
\left\lceil n/\delta\right\rceil-1\}\backslash\{j_0\}:\\
j\!~\equiv\!~\ell_0~\text{mod}~2}} \bar{h}_j(\bm
x)~\text{mod}~\overline{N}.\end{align*} Note that we have proved
that $\bar{h}_j(\bm x)<\overline{N}$ for each $j\in\{1,2,\cdots,
\left\lceil n/\delta\right\rceil-1\}$, so we actually have
\begin{align*}\bar{h}_{j_0}(\bm
x)=\bar{h}^{(\ell_0)}(\bm x)-\sum_{\substack{j\in\{1,2,\cdots,
\left\lceil n/\delta\right\rceil-1\}\backslash\{j_0\}:\\
j\!~\equiv\!~\ell_0~\text{mod}~2}} \bar{h}_j(\bm
x)~\text{mod}~\overline{N},\end{align*} where by Construction 3,
$$\bar{h}_{j_0}(\bm
x)=\left(\phi(c_{2,K_{j_0}}),\phi(c_{3,K_{j_0}}),\cdots,\phi(c_{\lceil\log
q\rceil,K_{j_0}})\right).$$ By observation ii) in Step 2, and by
Lemma \ref{lem-Bnry-burst-Sima}, $c_{i,[\lambda,\lambda']}$,
$i=2,\cdots,\lceil\log q\rceil$, can be recovered from
$\bar{h}_{j_0}(\bm x)$ and $d_{i,[\lambda,\lambda'-t']}$,
$i=2,\cdots,\lceil\log q\rceil$.

By the above discussions, $M_{\bm x}~($and so $\bm x)$ can be
uniquely recovered from $\bar{f}(\bm x)$ and any given $\bm
y\in\mathcal B_{\leq t}(\bm x)$.
\end{proof}

Let $\mathcal S_n$ be the set of all $q$-ary strings $\bm
x\in\mathbb Z_q^{n}$ such that the first row $c_{1,[n]}$ of
$M_{\bm{x}}$ is $(\bm p, \delta)$-dense, where
$M_{\bm{x}}=(c_{i,j})_{\lceil\log q\rceil\times n}$ is the matrix
representation of $\bm x$. Let
$$\mathcal
E^{\text{b}}_{\text{Den}}:\{0,1\}^{n-1}\rightarrow\mathcal
S^{\text{b}}_n$$ be the one-to-one mapping constructed as in
\eqref{E-B-Den}, where $\mathcal S^{\text{b}}_n$ is the set of all
$(\bm p, \delta)$-dense strings in $\{0,1\}^{n}$. By Lemma
\ref{lem-En-B2Q}, the mapping $\mathcal E^{\text{b}}_{\text{Den}}$
can be extended to a one-to-one mapping, denoted by $$\mathcal
E_{\text{Den}}:\mathbb Z_q^{n-1}\rightarrow\mathbb Z_q^{n},$$ such
that for each $\bm u\in\mathbb Z_q^{n-1}$ and $\bm x=\mathcal
E_{\text{Den}}(\bm u)$, if $M_{\bm{u}}=(b_{i,j})_{\lceil\log
q\rceil\times(n-1)}$ and $M_{\bm{x}}=(c_{i,j})_{\lceil\log
q\rceil\times n}$ are the matrix representation of $\bm u$ and
$\bm x$ respectively, then
$$c_{1,[n]}=\mathcal
E^{\text{b}}_{\text{Den}}(b_{1,[n-1]}).$$ Since $\mathcal
E^{\text{b}}_{\text{Den}}(b_{1,[n-1]})$ is $(\bm p,
\delta)$-dense, then for any $\bm u\in\mathbb Z_q^{n-1}$, we have
$\bm x=\mathcal E_{\text{Den}}(\bm u)\in\mathcal S_n$.

As in Section III, for each binary string $\bm a$, let $\mathcal
Q(\bm a)$ be the $q$-ary representation of $\bm a$. Note that
$\mathcal Q(\bm a)$ is a $q$-ary string of length $\left\lceil|\bm
a|/\lfloor\log q\rfloor\right\rceil$. Using the function $\bar f$
constructed in Construction 3 and the mapping $\mathcal
E_{\text{Den}}$, we can construct an encoding function of a
$q$-ary code capable of correcting a burst of at most $t$
deletions.

%\textbf{Encoding function}:
Let $\bar{\mathcal E}$ be a function defined on $\mathbb
Z_q^{n-1}$ of the form
\begin{align}\label{Btdel-enc-fun}
\bar{\mathcal E}(\bm u)=(\bm v, \bm v', \bm v''), ~\forall\!~\bm
u\in\mathbb Z_q^{n-1},
\end{align}
where $\bm v=\mathcal E_{\text{Den}}(\bm u)$, $\bm v'=\mathcal
E_{\text{Den}}(\mathcal Q(\bar f(\bm v)))$ and $\bm
v''=\text{Rep}_{t+1}(\mathcal Q(\bar f(\bm v')))$ such that
$\text{Rep}_{t+1}(\cdot)$ is the encoding function of the
$(t+1)$-fold repetition code.

\begin{thm}\label{thm-Btdel-enc}
The code $\bar{\mathcal C}=\{\bar{\mathcal E}(\bm u): \bm
u\in\mathbb Z_q^{n-1}\}$, where $\bar{\mathcal E}$ is given by
\eqref{Btdel-enc-fun}, is a $q$-ary code capable of correcting a
burst of at most $t$ deletions. The redundancy of $\bar{\mathcal
C}$ is at most $\log n+(8\log q+9)\log\log n+o(\log q\log\log
n)+\gamma_t$ in bits, where $\gamma_t$ is a constant depending
only on $t$.
\end{thm}
\begin{proof}
To prove that $\bar{\mathcal C}$ is capable of correcting a burst
of at most $t$ deletions, we adopt a similar strategy as in the
proof of Theorem \ref{thm-2del-enc}. Specifically, let
$$\bm x=\bar{\mathcal E}(\bm u)=(\bm v, \bm v', \bm
v'')\in\bar{\mathcal C},$$ where $\bm v=\mathcal
E_{\text{Den}}(\bm u)$, $\bm v'=\mathcal E_{\text{Den}}(\mathcal
Q(\bar f(\bm v)))$ and $\bm v''=\text{Rep}_{t+1}(\mathcal Q(\bar
f(\bm v')))$. Given any $\bm y\in\mathcal B_{\leq t}(\bm x)$,
denoting $t'=|\bm x|-|\bm y|$, then $t'\leq t$ and we have
$y_{[1,m_1-t']}\in\mathcal B_{\leq t}(\bm v)$,
$y_{[m_1,m_2-t']}\in\mathcal B_{\leq t}(\bm v')$ and
$y_{[m_2,m_3-t']}\in\mathcal B_{\leq t}(\bm v'')$, where $m_1=|\bm
v|, m_2=|(\bm v,\bm v')|$ and $m_3=|\bm x|=|(\bm v,\bm v',\bm
v'')|$. First, since $\bm v''=\text{Rep}_{t+1}(\mathcal Q(\bar
f(\bm v')))$ is a codeword of a $t$-deletion code, then $\mathcal
Q(\bar f(\bm v'))$, and so $\bar f(\bm v')$, can be recovered from
$y_{[m_2,m_3-t']}$. Further, by Theorem
\ref{thm-burst-del-sketch}, $\bm v'$ can be recovered from
$y_{[m_1,m_2-t']}$ and $\bar{f}(\bm v')$, and so $\bar{f}(\bm v)$
can be recovered from $\bm v'=\mathcal E_{\text{Den}}(\mathcal
Q(\bar f(\bm v)))$. Finally, by Theorem \ref{thm-burst-del-sketch}
again, $\bm v$ can be recovered from $y_{[1,m_1-t']}$ and
$\bar{f}(\bm v)$. Thus, $\bm x=(\bm v, \bm v', \bm v'')$ can be
recovered from any $\bm y\in\mathcal B_{\leq t}(\bm x)$, which
proves that $\bar{\mathcal C}$ is capable of correcting a burst of
at most $t$ deletions.

Since $\bm u\in\mathbb Z_q^{n-1}$ and $\bm v=\mathcal
E_{\text{Den}}(\bm u)\in\mathcal S_n\subseteq\mathbb Z_q^{n}$, so
$\bm v$ has $\log q$ bits redundancy. Moreover, by Theorem
\ref{thm-burst-del-sketch}, the length of $\bm v'$ is $$|\bm
v'|=\log n+8(\log q+1)\log\log n+o(\log q\log\log n)+\gamma_t$$
bits and the length of $\bm v''$ is \begin{align*}|\bm
v''|&=\log|\bm v'|+8(\log q+1)\log\log |\bm v'|+o(\log q\log\log
|\bm b'|)\\&~~~+\gamma_t\\&=\log\log n+\gamma_t+o(\log q\log\log
n)\end{align*} in bits. So the total redundancy of $\bm
c=\bar{\mathcal E}^{\text{b}}(\bm a)$ is
\begin{align*}\text{redundancy~of}~\bar{\mathcal
C}&=\log q+|\bm v'|+|\bm v''|\\&=\log n+(8\log q+9)\log\log
n\\&~~~+o(\log q\log\log n)+\gamma_t\end{align*} bits.
\end{proof}

\section{Conclusions}
We constructed systematic $q$-ary two-deletion correcting codes
and $q$-ary burst-deletion correcting codes, where $q\geq 2$ is an
even integer. For $q$-ary two-deletion codes, the redundancy of
our construction is $\log n$ higher than the best known explicit
binary codes and is lower than all existing explicit $q$-ary
codes. For $q$-ary burst-deletion codes, our construction is
scaling-optimal in redundancy.

It is also an interesting problem to generalize our constructions
to odd $q$. Another interesting problem is to construct explicit
$q$-ary $t$-deletion correcting codes that improves upon the state
of the art constructions in redundancy.

\appendices

\section{Proof of Lemma \ref{lem-2del-pstn}}

In this appendix, we prove Lemma \ref{lem-2del-pstn}. We first
need to prove the following lemma.

\begin{lem}\label{lem-2del-yc}
Suppose $\bm{c}\in\{0,1\}^n$ and $\{j_1,j_2\},
\{j_1',j_2'\}\subseteq[n]$ such that $j_1<j_2, j'_1<j'_2$ and
$j_1\leq j_1'$. If $c_{[n]\backslash\{j_1,j_2\}}
=c_{[n]\backslash\{j'_1,j'_2\}}$, then one of the following holds:
\begin{itemize}
 \item[1)] $c_{j_1},c_{j'_1}$ are in the same run of $\bm{c}$,
 and $c_{j_2},c_{j'_2}$ are in the same run of
 $\bm{c}$.
 \item[2)] There is an alternative substring $c_{[s_1,s_2]}$ of $\bm
 c$ of length $\geq 3$ such that $c_{j_1}$ and $c_{s_1}$ are in the
 same run of $\bm{c}$, $j_2=s_1+1$, $j_1'=s_2-1$ and
 $c_{j_2'}$ and $c_{s_2}$ are in the same run of $\bm{c}$.
\end{itemize}
\end{lem}
\begin{proof}
If $\{j_1,j_2\}=\{j'_1,j'_2\}$, then 1) holds. In the following,
we assume that $\{j_1,j_2\}\neq\{j'_1,j'_2\}$. Since
$c_{[n]\backslash\{j_1,j_2\}} =c_{[n]\backslash\{j'_1,j'_2\}}$, we
can denote $\bm b=c_{[n]\backslash\{j_1,j_2\}}
=c_{[n]\backslash\{j'_1,j'_2\}}.$ From $\bm
b=c_{[n]\backslash\{j_1,j_2\}}$ we can obtain
\begin{equation}\label{2-del-yc1}
b_i=\!\left\{\!\begin{aligned}
&c_i, ~~~~\text{for}~i\in[1,j_1-1],\\
&c_{i+1}, ~\text{for}~i\in[j_1,j_2-2],\\
&c_{i+2}, ~\text{for}~i\in[j_2-1,n-2].
\end{aligned}\right.
\end{equation}
Similarly, from $\bm b=c_{[n]\backslash\{j'_1,j'_2\}}$ we can
obtain
\begin{equation}\label{2-del-yc2}
b_i=\!\left\{\!\begin{aligned}
&c_i, ~~~~\text{for}~i\in[1,j'_1-1],\\
&c_{i+1}, ~\text{for}~i\in[j'_1,j'_2-2],\\
&c_{i+2}, ~\text{for}~i\in[j'_2-1,n-2].
\end{aligned}\right.
\end{equation}
Since $j_1<j_2$, $j'_1<j'_2$ and $j_1\leq j_1'$, then we can
divide our discussions into the following three cases:

Case 1: $j_1<j_2\leq j'_1<j'_2$. Combining \eqref{2-del-yc1} and
\eqref{2-del-yc2}, we can obtain
\begin{align}\label{eq1-2del-cmpr}
\bm b=~&c_1\cdots c_{j_1-1}c_{j_1+1}c_{j_1+2}\cdots
c_{j_2-1}c_{j_2+1}c_{j_2+2}c_{j_2+3}c_{j_2+4}\cdots
c_{j'_1-2}c_{j'_1-1}~c_{j'_1}~~c_{j'_1+1}c_{j'_1+2}c_{j'_1+3}\cdots
c_{j'_2-1}~c_{j'_2}~~c_{j'_2+1}\cdots c_n\nonumber\\
=~&c_1\cdots c_{j_1-1}~c_{j_1}~~c_{j_1+1}\cdots
c_{j_2-2}c_{j_2-1}~c_{j_2}~~c_{j_2+1}c_{j_2+2}\cdots
c_{j'_1-4}c_{j'_1-3}c_{j'_1-2}
c_{j'_1-1}c_{j'_1+1}c_{j'_1+2}\cdots
c_{j'_2-2}c_{j'_2-1}c_{j'_2+1}\cdots c_n.\end{align} We further
need to divide this case into the following two subcases.

Case 1.1: $j_1'-j_2$ is odd. By comparing the symbols of the
corresponding positions in $c_{[n]\backslash\{j_1,j_2\}}$ and
$c_{[n]\backslash\{j'_1,j'_2\}}$, we can obtain
\begin{align*}
c_{j_1}&=c_{j_1+1}=\cdots=c_{j_2-2}=c_{j_2-1}=c_{j_2+1}=c_{j_2+3}=\cdots\\&
=c_{j'_1-4}=c_{j'_1-2}=c_{j'_1}\end{align*} and
\begin{align*}
c_{j_2}&=c_{j_2+2}=c_{j_2+4}=\cdots=c_{j_1'-3}=c_{j_1'-1}=c_{j_1'+1}\\
&=c_{j_1'+2}=c_{j_1'+3}=\cdots=c_{j'_2}.\end{align*}

Case 1.2: $j_1'-j_2$ is even. By comparing the symbols of the
corresponding positions in $c_{[n]\backslash\{j_1,j_2\}}$ and
$c_{[n]\backslash\{j'_1,j'_2\}}$, we can obtain
\begin{align*}
c_{j_1}&=c_{j_1+1}=c_{j_1+2}=\cdots=c_{j_2-2}=c_{j_2-1}=c_{j_2+1}\\&
=c_{j_2+3}=\cdots=c_{j'_1-3}=c_{j'_1-1}=c_{j'_1+1}=c_{j'_1+2}\\
&=c_{j'_1+3}=\cdots=c_{j'_2-1}=c_{j'_2}\end{align*} and
\begin{align*}
c_{j_2}&=c_{j_2+2}=c_{j_2+4}=\cdots=c_{j_1'-4}=c_{j_1'-2}=c_{j_1'}\end{align*}

For the both subcases, we can see that
\begin{itemize}
\item If $c_{j_1}=c_{j_2}$, then we have
$c_{j_1}=c_{j_1+1}=\cdots=c_{j_2'}$, so $c_{j_1},c_{j_2},c_{j'_1}$
and $c_{j'_2}$ are in the same run of $\bm c$, which implies that
1) of Lemma \ref{lem-2del-yc} holds. \item If $c_{j_1}\neq
c_{j_2}$, then $c_{[j_2-1,j_1'+1]}$ is an alternative substring of
$\bm c$ of length $\geq 3$. By letting $s_1=j_2-1$ and
$s_2=j_1'+1$, we obtain 2) of Lemma
\ref{lem-2del-yc}.\end{itemize}

Example \ref{exam-2del-c1}) is an illustration of this case.

Case 2: $j_1\leq j'_1<j_2\leq j'_2$. In this case, combining
\eqref{2-del-yc1} and \eqref{2-del-yc2}, we can obtain
\begin{align}\label{eq2-2del-cmpr}
\bm b= &c_1\cdots c_{j_1-1}c_{j_1+1}c_{j_1+2}\cdots
c_{j'_1-1}~c_{j'_1}~~c_{j'_1+1}c_{j'_1+2}\cdots
c_{j_2-1}c_{j_2+1}c_{j_2+2}c_{j_2+3}\cdots
c_{j'_2-1}~c_{j'_2}~~c_{j'_2+1}\cdots c_n\nonumber\\
=&c_1\cdots c_{j_1-1}~c_{j_1}~~c_{j_1+1}\cdots
c_{j_1'-2}c_{j_1'-1}c_{j_1'+1}c_{j_1'+2}\cdots
c_{j_2-1}~c_{j_2}~~c_{j_2+1}c_{j_2+2}\cdots
c_{j'_2-2}c_{j'_2-1}c_{j'_2+1}\cdots c_n,\end{align} from which we
see that
$$c_{j_1}=c_{j_1+1}=\cdots=c_{j_1'}$$
and
$$c_{j_2}=c_{j_2+1}=\cdots=c_{j'_2}.$$
Hence, 1) of Lemma \ref{lem-2del-yc} holds.

Case 3: $j_1\leq j'_1<j'_2<j_2$. In this case, combining
\eqref{2-del-yc1} and \eqref{2-del-yc2}, we can obtain
\begin{align}\label{eq3-2del-cmpr}
\bm b= &c_1\cdots c_{j_1-1}c_{j_1+1}c_{j_1+2}\cdots
c_{j'_1-1}~c_{j'_1}~~c_{j'_1+1}c_{j'_1+2}\cdots
c_{j_2'-1}~c_{j_2'}~~c_{j_2'+1}c_{j_2'+2}\cdots
c_{j_2-1}c_{j_2+1}c_{j_2+2}\cdots c_n\nonumber\\
=&c_1\cdots c_{j_1-1}~c_{j_1}~~c_{j_1+1}\cdots
c_{j_1'-2}c_{j_1'-1}c_{j_1'+1}c_{j_1'+2}\cdots
c_{j_2'-1}c_{j_2'+1}c_{j_2'+2}c_{j_2'+3}\cdots
~c_{j_2}~~c_{j_2+1}c_{j_2+2}\cdots c_n,\end{align} from which we
can see that
$$c_{j_1}=c_{j_1+1}=\cdots=c_{j_1'}$$
and
$$c_{j_2'}=c_{j_2'+1}=\cdots=c_{j_2}.$$
Hence, 1) of Lemma \ref{lem-2del-yc} holds.
\end{proof}

\begin{exam}\label{exam-2del-c1}
To illustrate the Case 1 in the proof of Lemma \ref{lem-2del-yc},
let's consider the following two examples.
\begin{itemize}
 \item Consider $\bm c=01000101011110$. Let $j_1=3$,
 $j_2=6$, $j_1'=9$ and $j_2'=12$. Then $c_{[n]\backslash\{j_1,j_2\}}
 =c_{[n]\backslash\{j'_1,j'_2\}}=010001011110$ and $j_1'-j_2=9-6=3$
 is odd. We can find that
 $c_{j_1}=\cdots=c_{j_2-1}=c_{j_2+1}=c_{j_2+3}=\cdots=c_{j'_1}=0$
 and
 $c_{j_2}=c_{j_2+2}=\cdots=c_{j'_1-1}=c_{j'_1+1}=\cdots=c_{j'_2}=1$.
 \item Consider $\bm c=01000101010001$. Let $j_1=3$,
 $j_2=6$, $j_1'=10$ and $j_2'=12$. Then $c_{[n]\backslash\{j_1,j_2\}}
 =c_{[n]\backslash\{j'_1,j'_2\}}=010001010001$ and $j_1'-j_2=10-6=4$
 is even. We can find that
 $c_{j_1}=\cdots=c_{j_2-1}=c_{j_2+1}=c_{j_2+3}=\cdots=c_{j'_1-1}
 =c_{j'_1+1}=c_{j'_1+2}=\cdots=c_{j'_2}=0$ and
 $c_{j_2}=c_{j_2+2}=\cdots=c_{j'_1}=1$.
\end{itemize}
\end{exam}

Now, we can prove Lemma \ref{lem-2del-pstn}.

\begin{proof}[Proof of Lemma \ref{lem-2del-pstn}]
Suppose $\bm c\in\{0,1\}^n$ is regular and $\bm b\in\{0,1\}^{n-2}$
such that $\bm b$ can be obtained from $\bm c$ by deleting two
symbols of $\bm c$. Then we can always find two symbols of $\bm
c$, say $c_{j_1}$ and $c_{j_2}~(j_1<j_2)$, such that, $\bm
b=c_{[n]\backslash\{j_1,j_2\}}$.

First, suppose $c_{j_1}$ and $c_{j_2}$ are in the same run of $\bm
c$. Then for any $\{j_1',j_2'\}\subseteq[n]$ such that $\bm
b=c_{[n]\backslash\{j'_1,j'_2\}}$, it is easy to see that 2) of
Lemma \ref{lem-2del-yc} can't hold. $($Otherwise, there is an
alternative substring $c_{[s_1,s_2]}$ of $\bm c$ of length $\geq
3$ such that $c_{j_1}$, $c_{s_1}$ are in the same run of $\bm{c}$
and $j_2=s_1+1$, which implies that $c_{j_1}=c_{s_1}\neq
c_{s_1+1}=c_{j_2}$, which contradicts to the assumption that
$c_{j_1}$ and $c_{j_2}$ are in the same run of $\bm c$.$)$
Therefore, 1) of Lemma \ref{lem-2del-yc} must hold, which implies
that there is a run $c_{J}$ of $\bm c$, where $J\subseteq[n]$ is
an interval, such that $\bm b$ is obtained from $\bm c$ by
deleting two symbols in $c_{J}$. Since $\bm c\in\{0,1\}^n$ is
regular, by Remark \ref{rem-Regu-length}, the length of $c_{J}$ is
at most $d\log n<\rho=3d\log n$. Thus, 2) of Lemma
\ref{lem-2del-pstn} holds.

In the following, we suppose that $c_{j_1}$ and
$c_{j_2}~(j_1<j_2)$ are in two different runs of $\bm c$.
Specifically, suppose $j_1\in J=[i_1,i_2]\subseteq[n]$ and $j_2\in
J'=[i'_1,i'_2]\subseteq[n]$ such that $c_{J}$ and $c_{J'}$ are two
different runs of $\bm c$. Since $j_1<j_2$, then $$i_1\leq
i_2<i_1'\leq i_2'.$$ We need to consider the following two cases.

Case 1: $i_1'>i_2+1$. Since $j_2\in J'=[i'_1,i'_2]$, then $j_2\geq
i_1'>i_2+1$. For any $\{j_1',j_2'\}\subseteq[n]$ such that $\bm
b=c_{[n]\backslash\{j'_1,j'_2\}}$, it is easy to see that 2) of
Lemma \ref{lem-2del-yc} can't hold. $($Otherwise, there is an
alternative substring $c_{[s_1,s_2]}$ of $\bm c$ of length $\geq
3$ such that $c_{j_1}$, $c_{s_1}$ are in the same run of $\bm{c}$
and $j_2=s_1+1$, which implies that $s_1=i_2$ and
$j_2=s_1+1=i_2+1$, which contradicts to the fact that $j_2\geq
i_1'>i_2+1.)$ Therefore, 1) of Lemma \ref{lem-2del-yc} must hold,
which implies that $j_1'\in J=[i_1,i_2]$ and $j_2'\in
J'=[i'_1,i'_2]'$. Thus, 1) of Lemma \ref{lem-2del-pstn} holds.

Case 2: $i_1'=i_2+1$. We need to consider the following two
subcases.

Case 2.1: $|J|\geq 2$ and $|J'|\geq 2$. Then for any
$\{j_1',j_2'\}\subseteq[n]$ such that $\bm
b=c_{[n]\backslash\{j'_1,j'_2\}}$, it is easy to see that 2) of
Lemma \ref{lem-2del-yc} can't hold because no such alternative
substring $c_{[s_1,s_2]}$ of $\bm c$ can be found. Therefore, 1)
of Lemma \ref{lem-2del-yc} must hold, which implies that $j_1'\in
J=[i_1,i_2]$ and $j_2'\in J'=[i'_1,i'_2]'$. Thus, 1) of Lemma
\ref{lem-2del-pstn} holds.

Case 2.2: $|J|=1$ or $|J'|=1$. Without loss of generality, assume
$|J|=1$. Then $i_1=i_2$ and $c_{i_1-1}c_{i_1}c_{i_1+1}$ is an
alternative substring of $\bm c$. Let $c_{[\lambda_1,\lambda_2]}$
be the maximal alternative substring of $\bm c$ that contains
$c_{i_1-1}c_{i_1}c_{i_1+1}$, where
$[\lambda_1,\lambda_2]\subseteq[n]$ is an interval. Let
$c_{[\lambda_0,\lambda_1]}~($if $\lambda_1>1)$ and
$c_{[\lambda_2,\lambda_3]}~($if $\lambda_2<n)$ be two runs of $\bm
c$. For any $\{j_1',j_2'\}\subseteq[n]$ such that $\bm
b=c_{[n]\backslash\{j'_1,j'_2\}}$, by Lemma \ref{lem-2del-yc}, we
have $\{j_1',j_2'\}\subseteq[\lambda_0,\lambda_3]$. Since $\bm
c\in\{0,1\}^n$ is regular, by Remark \ref{rem-Regu-length}, the
length of the alternative substring $c_{[\lambda_1,\lambda_2]}$ of
$\bm c$ is at most $d\log n$, and the lengths of the runs
$c_{[\lambda_0,\lambda_1]}$, $c_{[\lambda_2,\lambda_3]}$ of $\bm
c$ are both at most $d\log n$. Hence, the length of
$c_{[\lambda_0,\lambda_3]}$ is at most $\rho=3d\log n$. Thus, 2)
of Lemma \ref{lem-2del-pstn} holds.

By the above discussions, we proved that exact one of the two
claims of Lemma \ref{lem-2del-pstn} holds.
\end{proof}

As an example, suppose $\bm c=011000101011110100$. We consider the
following cases.
\begin{itemize}
 \item If $\bm b=0110101011110100$, then $\bm b$ can be obtained
 from $\bm c$ by deleting two symbols in the run $c_{[4,6]}=000$.
 \item If $\bm b=0110010011110100$, then $\bm b$ can be obtained
 from $\bm c$ by deleting one symbol in the run $c_{[4,6]}=000$ and
 one symbol in the run $c_{[9,9]}=1$. This case is an example of
 Case 1 in the proof of Lemma \ref{lem-2del-pstn}.
 \item If $\bm b=0100101011110100$, then $\bm b$ can be obtained
 from $\bm c$ by deleting one symbol in the run $c_{[2,3]}=11$ and
 one symbol in the run $c_{[4,6]}=000$. This case is an example of
 Case 2.1 in the proof of Lemma \ref{lem-2del-pstn}.
 \item If $\bm b=0110001011110100$, then $\bm b$ can be obtained
 from $\bm c$ by deleting two symbols in the substring
 $c_{[4,14]}=00010101111$, which may be any of the following cases:
  i) one symbol in the run
 $c_{[4,6]}=000$ and the symbol $c_{7}=1$; ii) the symbols
 $c_i,c_{i+1}$ for $i\in\{7,8,9\}$; iii) one symbol in the run
 $c_{[11,14]}=1111$ and the symbol $c_{10}=0$. This case is an example of
 Case 2.2 in the proof of Lemma \ref{lem-2del-pstn}.
\end{itemize}

\end{document}